\documentclass[pra,twocolumn,showpacs]{revtex4-2}

\usepackage{amsmath}
\usepackage{latexsym}
\usepackage{amssymb}
\usepackage{graphicx}
\usepackage[colorlinks=true, citecolor=blue, urlcolor=blue]{hyperref}
\usepackage{float}
\usepackage{amsfonts}
\usepackage{textcomp}
\usepackage{mathpazo}
\usepackage{comment}
\usepackage{xr}

\usepackage{bbm}

\usepackage{xcolor}
\definecolor{myurlcolor}{rgb}{0,0,0.4}
\definecolor{mycitecolor}{rgb}{0,0.5,0}
\definecolor{myrefcolor}{rgb}{0.5,0,0}
\usepackage{hyperref}
\hypersetup{colorlinks,
linkcolor=myrefcolor,
citecolor=mycitecolor,
urlcolor=myurlcolor}

\usepackage[draft]{fixme}
\usepackage{amsmath,bbm}%,mathrsfs}
\usepackage{graphicx}
\usepackage{amsfonts}
\usepackage{amssymb}
\usepackage{amsmath, amssymb, amsthm,verbatim,graphicx,bbm}
\usepackage{mathrsfs}
\usepackage{color,xcolor,longtable}

\usepackage{xr}
\makeatletter
\newcommand*{\addFileDependency}[1]{
  \typeout{(#1)}
  \@addtofilelist{#1}
  \IfFileExists{#1}{}{\typeout{No file #1.}}
}
\makeatother

\newcommand*{\myexternaldocument}[1]{
    \externaldocument{#1}
    \addFileDependency{#1.tex}
    \addFileDependency{#1.aux}
}
\myexternaldocument{manuscript_PRL}
\newcommand{\beq}[0]{\begin{equation}}
\newcommand{\eeq}[0]{\end{equation}}

\newcommand{\one}{\leavevmode\hbox{\small1\normalsize\kern-.33em1}}

\def\be{\begin{equation}}
\def\ee{\end{equation}}
\def\ben{\begin{eqnarray}}
\def\een{\end{eqnarray}}
\def\eea{\end{array}}
\def\bea{\begin{array}}

\newcommand{\Tr}[1]{\mathrm{Tr}#1}
\newcommand{\bei}{\begin{itemize}}
\newcommand{\eei}{\end{itemize}}
\newcommand{\ket}[1]{|#1\rangle}
\newcommand{\bra}[1]{\langle#1|}

\newcommand{\proj}[1]{\ket{#1}\!\!\bra{#1}}
\newcommand{\braket}[2]{\langle{#1}|{#2}\rangle}

\newcommand{\I}{\mathbbm{1}}

\newcommand{\p}{\Vec{p}}

\renewcommand{\emph}[1]{\textbf{#1}}

\newcommand{\eqnref}[1]{(\ref{#1})}
\newcommand{\figref}[1]{Fig.~\ref{#1}}

\makeatletter
\newtheorem*{rep@theorem}{\rep@title}
\newcommand{\newreptheorem}[2]{%
\newenvironment{rep#1}[1]{%
 \def\rep@title{#2 \ref{##1}}%
 \begin{rep@theorem}}%
 {\end{rep@theorem}}}
\makeatother

\theoremstyle{plain}
\newtheorem{thm}{Theorem}%[section]
\newtheorem*{thm*}{Theorem}
\newreptheorem{thm}{Theorem}

\newtheorem{fakt}{Fact}

\newtheorem{assu}[thm]{Assumption}

\theoremstyle{definition}

\theoremstyle{remark}

\usepackage[T1]{fontenc}

\begin{document}

\title{Topologically noise robust network steering without inputs}
\author{Dhruv Baheti}
\email{dhruvbaheti27@gmail.com}
\affiliation{Laboratoire d’Information Quantique, Université libre de Bruxelles (ULB), Av. F. D. Roosevelt 50, 1050 Bruxelles, Belgium}
\affiliation{Department of Physics, Indian Institute of Technology Kharagpur,
Kharagpur 721302, India}
\author{Shubhayan Sarkar}
\email{shubhayan.sarkar@ug.edu.pl}
\affiliation{Laboratoire d’Information Quantique, Université libre de Bruxelles (ULB), Av. F. D. Roosevelt 50, 1050 Bruxelles, Belgium}
\affiliation{Institute of Informatics, Faculty of Mathematics, Physics and Informatics,
University of Gdansk, Wita Stwosza 57, 80-308 Gdansk, Poland}

%%%%%%%%%%%%%%%%%%%%%%%%%%%%%%%%%%%%%%%%%%%%%%%%%%%%%%%%%%%%%%%%%%%
\begin{abstract}	
Quantum networks with independent sources allow observing quantum nonlocality or steering with just a single measurement per node of the network, or without any inputs. Inspired by the recently introduced notion of swap-steering, we consider here the triangle network scenario without inputs, where one of the parties is trusted to perform a well-calibrated measurement. In this scenario, we first propose a linear witness to detect triangle network swap-steering. Then, by using the correlations that achieve the maximum value of this inequality, and assuming that all the sources are the same, we can self-test the state generated by the sources and the measurements of the untrusted party. We then extend this framework to ring networks with an arbitrary number of nodes with one of them being trusted. Interestingly, this is the first example of a topologically robust, that is, one can observe steerability without assuming the network structure of the network, as well as noise-robust quantum advantage in a network. Additionally, by allowing the trusted party to perform tomography of their subsystems, we demonstrate that every bipartite entangled state will result in swap-steerable correlations in the ring network. For this purpose, we construct linear witnesses to detect ring network swap-steering corresponding to every bipartite entangled state. 
\end{abstract}

%%%%%%%%%%%%%%%%%%%%%%%%%%%%%%%%%%%%%%%%%%%%%%%%%%%%%%%%%%%%%%%%%%%

\maketitle

\section{Introduction} Quantum non-locality is one of the most remarkable features of quantum mechanics that defies our classical intuitions about the world. It refers to the ability of separated quantum systems to set-up correlations that are incompatible with assumptions of local causality. This quantum property was first conceptualized in the celebrated work of Einstein, Podolsky and Rosen \cite{EPR}, following which Bell in 1964 \cite{Bell, Bell66} formulated a theoretical test to observe non-classical correlations.  It was then experimentally verified \cite{Bellexp1, Bellexp2, Bellexp3, Bellexp4} and is now recognized as a fundamental aspect of quantum mechanics. Quantum non-locality has now found applications in the study of a wide variety of topics such as cryptography, quantum teleportation, quantum communication, and quantum computing (refer to \cite{NonlocalityReview} for a review). Quantum non-locality can be re-formulated, with an additional assumption, into quantum steering, where an observer is able to remotely influence a distant quantum system. Quantum steering was initially introduced by Schrodinger \cite{Schrod} and was rigorously formulated in \cite{Wiseman}. The additional assumption is that one of the parties is trusted to perform fixed quantum measurements.

To observe quantum nonlocality or steering in the standard scenario, each of the involved parties in the experiment must have at least two incompatible measurements as inputs. Interestingly, multiple networked quantum sources allow non-classicality without the need for incompatible measurements. The framework to witness quantum nonlocality in networks was introduced in \cite{pironio1,Pironio22, pironio2, Fritz}. However, it was first noted in \cite{Pironio22} and then in \cite{Fritz} that considering independent sources between non-communicating parties allows one to observe quantum nonlocality with a single fixed measurement for every party. %This form of non-locality is called genuine network non-locality, as it is intrinsic to the networked structure of the experiment. This notion was introduced in  and 
The first example of a genuinely network non-local correlation, that is, the correlations that are intrinsic to the network structure and thus can not be simulated using smaller networks, was presented in \cite{renou1} in the case of three nodes with three independent sources, also referred to as the triangle network.%, that is quantum nonlocality that is intrinsic to the networked structure.

%The results of \cite{renou1} were then extended for a family of network structures in \cite{Renou_2022} and further extended to even more structures in \cite{boreiri2024topologicallyrobustquantumnetwork} by trusting local network topology. On the other hand, 
The notion of quantum steering in networks was introduced in \cite{netstee} where some nodes in the network were assumed to be trusted in the sense that local subsystems at these nodes could be precisely known, or equivalently, the parties at the trusted nodes could perform a tomographically complete set of measurements on their incoming subsystems. The minimal scenario where one could observe quantum steering in networks without inputs was introduced in \cite{Sarkar2024networkquantum} and is known as swap-steering. The scenario comprised of two sources and two parties without inputs, where one of the parties is trusted to perform a single bell measurement. Unlike the case in \cite{renou1}, to demonstrate swap-steering, a linear witness was constructed in \cite{Sarkar2024networkquantum}. This makes swap-steering much easier to observe in experiments. Moreover, it was also recently demonstrated that every entangled state \cite{sarkar2024_every} and every entangled measurement \cite{entmeauniv} can be used to observe swap-steering.

In this work, we generalize the notion of swap-steering to ring networks with arbitrary number of nodes, primarily focusing on the triangle network where only a single node of the network is trusted. In this scenario, assuming the trusted node to perform the Bell measurement, we construct linear witnesses to observe swap-steering in any ring network. For the triangle network, with an additional assumption of identical sources, we introduce a self-testing scheme for both the quantum states and measurements. Moreover, we show that swap-steering in the $n-$ring network is topologically robust for any network structure among the untrusted parties. Thus, one can observe a quantum advantage even when the untrusted parties might communicate with each other, which is impossible to observe in standard multipartite steering scenarios. 

A similar result was obtained in \cite{boreiri2024topologicallyrobustquantumnetwork} where none of the parties is trusted. However, it is based on observing particular probabilities in the network and is thus not ideal for experimental implementation as it is highly susceptible to noise \cite{boreiri2024noiserobustproofsquantumnetwork,Kriv_chy_2020}. Moreover, these result requires the assumption of independence and identical distribution (iid) \cite{weilenmann}. Our method is based on deriving linear witness for the arbitrary network to show nonclassicality in the network with a considerable gap between the quantum and the swap-unsteerable values making it ideal for observing in experiments and also does not require the iid assumption. Consequently, we provide the first example of a topologically as well as noise-robust quantum advantage in a network. Finally, by enabling the trusted party to perform tomography on their subsystems, we show that every bipartite entangled state can generate swap-steerable correlations in the ring network. To achieve this, we construct linear witnesses that can detect ring network swap-steering for all bipartite entangled states.

\section{Triangle network swap-steering} 
%The arrangement of the sources and measurement nodes for the triangle is depicted in Fig. \ref{fig1}. 
Let us begin by describing the triangle network steering scenario. Each party (or node) in the network receives a subsystem from each source connected to it on which the party performs a single four-outcome measurement. The outcomes are denoted as $a_n$ for the $n^{th}$ node such that $a_n=0,1,2,3$ and $n=1,2,3$. One of the measurement parties is trusted here implying that the measurement performed on it's subsystem is known (see Fig. \ref{fig1}). Here we trust that $A_1$ performs the measurement corresponding to the Bell basis given by $ N_{A_1}=\{\proj{\phi_{+}},\proj{\phi_{-}},\proj{\psi_{+}},\proj{\psi_{-}}\} $ where
\begin{eqnarray}\label{Amea1}
    \ket{\phi_{\pm}}&=&\frac{1}{\sqrt{2}}\left(\ket{00}\pm\ket{11}\right)\nonumber\\
    \ket{\psi_{\pm}}&=&\frac{1}{\sqrt{2}}\left(\ket{01}\pm\ket{10}\right).
\end{eqnarray}
The experiment is then repeated enough times to construct the joint probability distribution (correlations) $\vec{p}=\{p(a_{1},a_{2},a_{3})\}=\{p(\{a\}_3)\}$ where $\{a\}_n$ denotes strings of $n$ outputs $(a_{1},a_{2},\ldots,a_{n})$ where $a_i=0,1,2,3$. These probabilities can be computed in quantum theory for the triangle network as
\begin{eqnarray}
p\left(\{a\}_3\right)=\Tr{\left[ \left( N_{A_1}^{a_{1}}M_{A_2}^{a_{2}}M_{A_3}^{a_{3}} \right) \rho_{A_1^2A_2^1}\rho_{A_2^2A_3^1}\rho_{A_3^2A_1^1} \right]}. 
\end{eqnarray}
Here $A_i^j$ denotes the $j^{th}$ subsystem received by the $i^{th}$ node and $N_{A_1}^{a_{1}}/M^{a_i}_{A_{i}}$ denotes trusted/untrusted measurement elements of $A_i$ corresponding to the outcome $a_i$, are positive and $\sum_{a_i}M_{A_i}^{a_{i}}=\sum_{a_1}N_{A_1}^{a_{1}}=\I$. 
It is important to recall that the measurement parties can not communicate during the experiment. % or equivalently the correlations $\vec{p}$ need to be no-signalling. 
%For a note, the simplest network with two sources and two nodes was considered in \cite{Sarkar2024networkquantum}.
%%
%%%

{\it{SOHS models---}}
The notion of unsteerability in the simplest quantum network was introduced in \cite{Sarkar2024networkquantum} and is based on two assumptions, namely outcome-independence and separable quantum sources. Any hidden state model satisfying the above two assumptions is termed separable outcome-independent hidden state (SOHS) model. %If the collection of states at the trusted node for different outcomes of the untrusted measurements, following measurement at all other nodes is incompatible with any SOHS model, then it is considered a swap-steerable state. Similarly, incompatible correlations are termed swap-steering correlations. 
Let us here modify the assumptions of \cite{Sarkar2024networkquantum} for the triangle network with one trusted node. For this purpose, let us begin by considering the most general model to reproduce any correlation $p\left(\{a\}_3\right)$ in the triangle network [see Fig. \ref{fig1}] given by
\begin{eqnarray}
    p\left(\{a\}_3\right)=\sum_{\lambda_1,\lambda_2,\lambda_3}p(a_1,a_2,a_3|\lambda_1,\lambda_2,\lambda_3)p(\lambda_1,\lambda_2,\lambda_3)\ \ 
\end{eqnarray}
which can be expressed using the Bayes' rule as
\begin{eqnarray}
\sum_{\lambda_1,\lambda_2,\lambda_3}p(a_1|\lambda_1,\lambda_3)p(a_2,a_3|\lambda_1,\lambda_2,\lambda_3,a_1)p(\lambda_1,\lambda_2,\lambda_3).\ \ 
\end{eqnarray}
Here, the hidden variable $\lambda$ is in general any physical quantity that can generate correlations, for instance, they can be some classical variables, quantum states or even states in generalised probabilistic theories. Notice that in the above formula, we directly take $p(a_1|\lambda_1,\lambda_2,\lambda_3)=p(a_1|\lambda_1,\lambda_3)$ as $A_1$ does not recieve the hidden variable $\lambda_2$ from the source $S_2$. As $A_1$ is trusted to perform a particular quantum measurement on some quantum state, we impose in the above model that the source sends some quantum states to $A_1$ based on the variables $\lambda_1,\lambda_3$ and thus $ p\left(\{a\}_3\right)$ from the above formula can be expressed as 
\begin{equation}\label{eq55}
\sum_{\lambda_1,\lambda_2,\lambda_3}\Tr(N^{a_1}_{A_1}\rho_{\lambda_1,\lambda_3})p(a_2,a_3|\lambda_1,\lambda_2,\lambda_3,a_1)p(\lambda_1,\lambda_2,\lambda_3).\quad
\end{equation}

\begin{figure}[t]
\includegraphics[width=\linewidth]{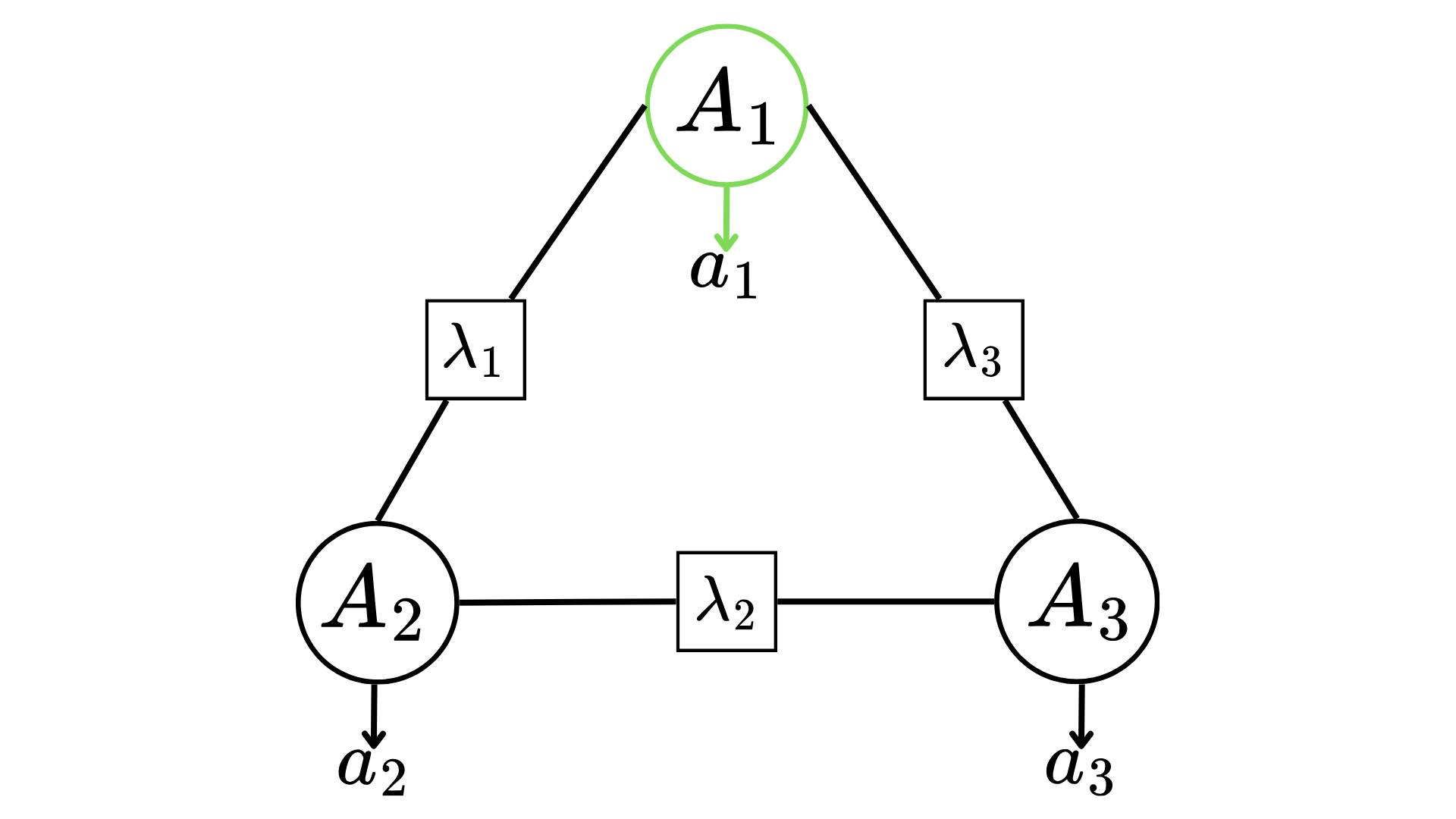}
    \caption{Triangle swap-steering scenario. Each node receives two subsystems from the two sources connected to them, on which a single four-outcome measurement is performed. Node $A_1$ is trusted here, meaning a Bell-basis measurement is performed at this node. No communication is allowed between the nodes during the experiment. By repeatedly performing the experiment, one obtains the joint probability distribution $\{p(a_1,a_2,a_3)\}$.}
    \label{fig1}
\end{figure} 

Let us now further restrict the above model to the case when the correlations can be explained locally, that is, via some local states on the trusted end and some local variables on the untrusted side. Similar to \cite{Sarkar2024networkquantum}, the notion of locality in the triangle network [see Fig. \ref{fig1}] is defined via the following two assumptions:

\begin{assu}[Outcome-independence]\label{ass1} The outcomes of parties are independent of each other if one has access to the hidden variables $\lambda_i$ from all the sources they are connected to.
\end{assu}
In the scenario considered in this work, we formulate the above statement for the set of all outcomes and hidden variables $\{a\}_3,\{\lambda\}$ respectively, as
\begin{eqnarray}\label{oi-eq}
p(a_{i}|\{\lambda\},\{a\}_{2/i})=p(a_{i}|\{\lambda_i\})\quad \forall i
\end{eqnarray}
where $\{\lambda_{i}\}$ is the set of all sources connected to node $i$ and $\{a\}_{2/i}$ is the set of outcomes of all the other nodes. 
\begin{assu}[Separate sources]\label{ass2} Sources $S_i$ $(i=1,2,3)$ distribute hidden variables to the nodes they are connected to such that the joint distribution of the hidden variables is a convex combination of product distributions. 
\end{assu}

%Within quantum theory, this would imply that if the sources generate the states $\rho_i$, then the joint state generated by the sources $\rho_{12}$ will be given by $\rho_{12}=\rho_{1}\otimes\rho_2$.
The above statement is formulated mathematically as 
\begin{eqnarray}\label{si-eq}
    p(\{\lambda\}) = \sum_{k}\prod_{j=1}^{3} p^{(k)} p(\lambda_j^{(k)}).
\end{eqnarray}
Note that $p^{(k)}$ is a distribution that specifies the probability of preparing different product distributions. Thus, here we allow for the sources to distribute convex sums of independent hidden variables. Since at the trusted end, the correlations are generated via some quantum states, for each variable $\lambda_i^{(k)}$, $A_1$ receives some quantum state $\rho_{\lambda_i}^{(k)}$. As independently prepared states in quantum theory are described via the tensor product rule, we have that the joint state at $A_1$ is $\rho_{\lambda_1}^{(k)}\otimes \rho_{\lambda_3}^{(k)}$ that are sent with a distribution $\{p^{(k)}\}$.

The assumption of separable sources %, or can be equivalently called, classically correlated sources, 
stems from the assumption of independent sources in quantum networks \cite{Fritz, pironio1,
Marco}. However, it is a weaker assumption than that of independent sources and allows the sources to be classically correlated. We will later on further weaken this assumption to the case when only the sources connected to the trusted party are separable with no other restriction imposed on other sources. Particularly, for the trusted node $a_1$, we have that 
%\begin{eqnarray}\label{eq-5}
 %   p(a_1|\{\lambda^{(k)}\}, \{a\}_{2/1}) = p(a_1|\lambda_1^{(k)}, \lambda_3^{(k)}) =  p(a_1|\rho_{\lambda_1^{(k)}}\otimes\rho_{\lambda_3^{(k)}})\nonumber \\ = \Tr{[N^{a_1}_{A_1}\rho_{\lambda_1^{(k)}}\otimes\rho_{\lambda_3^{(k)}}]}.\qquad
%\end{eqnarray}
\begin{eqnarray}\label{eq-5}
p(a_1|\rho_{\lambda_1^{(k)}}\otimes\rho_{\lambda_3^{(k)}})= \Tr{[N^{a_1}_{A_1}\rho_{\lambda_1^{(k)}}\otimes\rho_{\lambda_3^{(k)}}]}.
\end{eqnarray}

Now, given that the sources $S_i$ for $i=1,2,3$ generate some (for now hidden) states $\lambda_i^{(k)}$, we can always express the probability $p(\{a\}_3)$ for any $\{a\}_3$ as
\begin{eqnarray}
    p(\{a\}_3)=\sum_{k}\sum_{\{\lambda^{(k)}\}} p^{(k)} p(\{\lambda^{(k)}\}) p(\{a\}_3|\{\lambda^{(k)}\}).
\end{eqnarray}
Using Bayes rule and outcome-independence \eqref{oi-eq}, we have
\begin{eqnarray}
      p(\{a\}_3)=\sum_{k}\sum_{\{\lambda\}} p^{(k)} \prod_{i=1}^{3} p(\{\lambda^{(k)}\}) p(a_{i}|\{\lambda_{i}^{(k)}\}).
\end{eqnarray}
Assuming separable sources [see Eqs. \eqref{si-eq} and \eqref{eq-5}] we obtain that $ p(\{a\}_3)$ for any SOHS model is given by
\begin{equation}\label{SOHS}
   \sum_{k}\sum_{\{\lambda\}}\prod_{j=1}^{3} p^{(k)} p(\lambda_{j}^{k})\Tr{[N^{a_1}_{A_1}\rho_{\lambda_1^{k}}\otimes\rho_{\lambda_3^{k}}]}\prod_{i=2}^{3} p(a_{i}|\{\lambda_{i}^{(k)}\}).
\end{equation}
%Separable outcome-independent hidden state correlations $\vec{p}_{SOHS}$ admit the form \eqref{SOHS}. 

To witness swap-steering in the triangle network, a functional $W$ can be constructed which depends on $\vec{p}$ as
\begin{eqnarray}
W(\vec{p})=\sum_{\{a\}_3}c_{\{a\}_3}p(\{a\}_3) \leq \beta_{SOHS}
\end{eqnarray}
where $c_{a_1,a_2,a_3}$ are real coefficients and $\beta_{SOHS}$ denotes the maximum value attainable using assemblages admitting the SOHS model \eqref{SOHS}. %This will be referred to as the local bound. 
For the purpose of this article, we consider only linear functionals of $\vec{p}$ and thus, we can drop the index $k$ in Eq. \eqref{SOHS} and optimize over individual preparations as no convex combination can increase the value of linear functionals. Consequently, from now on for any SOHS model,
\begin{equation}\label{eq13}
   p(\{a\}_3)= \sum_{\{\lambda\}}\prod_{j=1}^{3} p(\lambda_{j})\Tr{[N^{a_1}\rho_{\lambda_1}\otimes\rho_{\lambda_3}]}\prod_{i=2}^{3} p(a_{i}|\{\lambda_{i}\}).
\end{equation}
%Notice that one would get the same expression as 

If correlations in the triangle network do not admit a SOHS model, then the correlations are swap-steerable. Now, from the expression \eqref{eq13}, it is clear that if the local state $\rho_{\lambda_1,\lambda_3}$ is not separable, then such correlations do not admit a SOHS model. This is possible if the joint state prepared by the sources $S_1,S_3$ already sends some entangled state to the trusted party and thus the assumption of separable sources is violated. However, based on the current understanding of the physical world, it should always be possible to prepare two sources that do not produce entangled joint systems. Considering that the sources are separable, this implies that outcome-independence is violated which is detected by the trusted party as his local state is entangled due to the measurement of other parties. Thus, other parties could steer or change the trusted party's local state from separable to entangled using local operations on spatially separated labs. 

Comparing it to the notion of network quantum steering \cite{netstee}, we weaken the assumption of independent sources to separable sources. Moreover, Ref. \cite{netstee} considers that the trusted party can perform a tomography of the received subsystem and then reconstruct the collection of post-measured states known as assemblage. Here, we restrict that the trusted party also performs a single measurement like the other parties and does not have complete information about the assemblage.
Let us now find a functional that can detect swap-steering in the triangle network.

\subsection{Swap-Steering}
Consider the following three party swap-steering functional 
\begin{eqnarray}\label{steering-3}
W_{3} = p(0,0,0) + p(0,1,1) + p(0,2,2) + p(0,3,3) \quad \nonumber \\ + p(1,0,1)  + p(1,1,0) + p(1,2,3) + p(1,3,2) \quad \nonumber \\ + p(2,0,2) + p(2,2,0) + p(2,1,3)+ p(2,3,1) \quad \nonumber \\ +  p(3,0,3) + p(3,3,0) + p(3,1,2) + p(3,2,1).  \quad
\end{eqnarray}
%Consider again the probability distribution $\vec{p}$. These correlations are not steerable from Bob to Alice, given that the sources $S_1, S_2$ are independent, if there exists a local hidden state (SOHS) model such that
%\begin{eqnarray}\label{SOHS2}
%p(a,b)=\sum_{\lambda_1,\lambda_2} \ p(\lambda_1)p(\lambda_2)p(a|\rho_{\lambda_1},\rho_{\lambda_2})p(b|\lambda_1,\lambda_2)
%\end{eqnarray}
%for all $a,b$. Here $\lambda_1,\lambda_2$ denote the hidden variables that are generated by the sources $S_1,S_2$ respectively and $l_1,l_2=0,1$.
%A convenient way to express the above witness would be to change the labelling of the outcomes as follows: $0 \rightarrow \alpha$, $1 \rightarrow \beta$, $2 \rightarrow -\alpha$ and $3 \rightarrow -\beta$. Now, we introduce the set $\mathcal{P}_{W_3} = \{\{a\}_3|a_1a_2a_3 = \alpha^{3-2j}\beta^{2j}\}$, where $j = 0,1$. Then,
%\begin{eqnarray}
 %   W_3 = \sum_{\mathcal{P}_{W_3}} p(a_1,a_2,a_3).
%\end{eqnarray}
%This compact notation will be particularly useful in the later sections when we consider ring networks with arbitrary number of nodes.

The above functional can be represented in a condensed way as $W_3=\sum_{\{a\}_3}c_{\{a\}_3}p(\{a\}_3)$ where
\begin{equation}
   c_{\{a\}_3}=\begin{cases}
       1\ \mathrm{if}\ a_1-a_2-(-1)^{a_2}a_3\ \mathrm{mod}\ 4 = 0\\
       0\ \mathrm{otherwise}.
   \end{cases}
\end{equation}

Recall here that Alice is trusted and performs the Bell-basis measurement with elements given in \eqref{Amea1}. Let us now find the maximum value that can be achieved using correlations that admit a SOHS model \eqref{SOHS}.

\begin{fakt}\label{fact1}
    The maximum value $\beta_{SOHS}$ of the three party swap-steering functional $W_3$ \eqref{steering-3} that can be obtained using a SOHS model is $\frac{1}{2}$.
\end{fakt}
The proof of the above is stated in Fact 1 of the Appendix.
\iffalse
\begin{proof}
    For a local hidden state model three party swap-steering functional $W_3$ is given by 
    \begin{eqnarray}
        W_3 = \sum_{A_3}\sum_{\lambda_1,\lambda_2,\lambda_3} p(\lambda_{1})p(\lambda_{2})p(\lambda_{3})\dots\quad\nonumber \\ \dots\Tr{[M^{a}\rho_{\lambda_1}\otimes\rho_{\lambda_n}]}p(a_{2}|\lambda_{1},\lambda_{2})p(a_{3}|\lambda_{2},\lambda_{3})
    \end{eqnarray}
    Maximizing $\Tr{[M^{a}\rho_{\lambda_1}\otimes\rho_{\lambda_n}]}$ over all possible qubit states $\rho_{\lambda_1},\rho_{\lambda_1}$
    \begin{eqnarray}
        \Tr{[M^{a}\rho_{\lambda_1}\otimes\rho_{\lambda_n}]} \leq \frac{1}{2}
    \end{eqnarray}
    For any of the projectors $M^{a}$ of the trusted bell basis measurement \eqref{Amea1}. We thus obtain the inequality 
    \begin{eqnarray}
        W_3 \leq \frac{1}{2}\sum_{A_3}\sum_{\lambda_1,\lambda_2,\lambda_3} p(\lambda_{1})p(\lambda_{2})p(\lambda_{3})\dots\quad\nonumber \\ \dots p(a_{2}|\lambda_{1},\lambda_{2})p(a_{3}|\lambda_{2},\lambda_{3})
    \end{eqnarray}
    Observe that since the expression is no longer a function of $a_1$, the sum over $\{W_3\}$ is simply a sum over all $(a_2,a_3)$. Using the fact that $\sum_{a_i}p(a_i|\{\lambda_i\}) = 1$ and $\sum_{\lambda_i}p(\lambda_i) = 1$, we obtain the bound
    \begin{eqnarray}
        W_3 \leq \beta_{SOHS} = \frac{1}{2}
    \end{eqnarray}
    This concludes the proof.
\end{proof}
\fi
%
%
%
Now, consider the following quantum experiment. Each source prepares the state $\ket{\psi_i}=\ket{\phi_+}_{A_i^2A_{i+1}^1}$ and each node measures in the same basis as the trusted node, that is, $M_{A_{i}}=\{\proj{\phi_{+}},\proj{\phi_{-}},\proj{\psi_{+}},\proj{\psi_{-}}\}_{A_{i}^1A_i^2}$ where the corresponding states are given in \eqref{Amea1}. Using these states and measurements, one can simply evaluate the steering functional $W_3$ \eqref{steering-3} to get the value $1$. Notice that this is also the algebraic bound of $W_3$. Thus, we conclude that the experiment described generates the quantum correlation that maximally violates the three-party swap-steering inequality. It is also interesting to observe here that with the given quantum states and measurements, the post-measured state at the trusted node $A_1$ is one of the four bell states \eqref{Amea1}. %Thus they do not admit a decomposition of the form $\sum_{i}p^{i}\rho_1^i\otimes\rho_2^i$. 
Consequently, the state with the trusted party is steered from a separable state to an entangled one due to the measurement of the other parties and is detected via the swap-steering witness.

{\subsection{Necessary conditions for swap-steering}}
Let us now find some necessary conditions to observe swap-steering in the triangle network.
We show that to violate the SOHS bound, each of the sources and measurements must be entangled. For the case of measurements, we say it is entangled if at least one of the basis elements is entangled. Conversely, we call a measurement separable if all its elements are separable.
\begin{fakt}\label{fact3}
    If any one of the states generated by the sources or the measurement of any untrusted node is separable, then the maximum value that can be obtained of the swap-steering functional $W_3$ \eqref{steering-3} is $\frac{1}{2}$.
\end{fakt}
 %Here, when we say that the measurement basis is separable if each of the basis elements is separable. 
 The proof of the above statement is stated in Fact 3 of the Appendix. The key insight here is that in the experiment described above for the maximal violation of the SOHS bound, the sources set up a chain of entanglement between the nodes which perform entangled measurements to swap the entanglement between the pairs of systems, resulting in an entangled state at the trusted node post measurement at all the other nodes. However, when any of the states or measurements are separable, there is a break in the link, thus, entanglement cannot be swapped across the network to the trusted node's system. Thus, steering functional is bounded by $\frac{1}{2}$.

{\subsection{Weak self-testing in the triangle network}} 
Self-testing is a theoretical framework that allows for the certification of quantum states and measurements based purely on observed data, without needing detailed knowledge about the inner workings of the devices involved. By analyzing the correlations between measurement outcomes, self-testing guarantees that a quantum system behaves in a specific way, such as using particular entangled states, like Bell states, or performing certain quantum measurements. There are numerous self-testing schemes that aim to certify quantum states and measurements [for instance see Refs. \cite{Scarani,Yang,Bamps,All,chainedBell, sarkar, Coladangelo_2017, sarkar4}]. Furthermore, self-testing of quantum states and measurements in quantum networks was recently explored in Refs. \cite{Marco, NLWEsupic, JW2, Allst1, sarkar2023_UPB, supic4,Sarkar_2024_opera,sarkar2024universal}. In particular \cite{Allst1} certifies every pure entangled state and \cite{sarkar2024universal} certifies every quantum state. The idea of self-testing of quantum states and measurements has also been extended to scenarios involving trusted parties also called one-sided device-independent certification [for instance see Refs. \cite{Supic, Alex, sarkar6, sarkar12, Bharti, sarkar11,sarkar2024_GME}]. These schemes are in general more robust to noise and thus easier to implement than the former case. Recently, in \cite{Sarkar2024networkquantum} one-sided device-independent certification has been extended to the case of quantum networks without inputs.

Consider again the triangle scenario described above in Fig. \ref{fig1}. %Let us assume that all the probabilities in the witness $W_3$ \eqref{steering-3} are equal to $1/16$. Notice that this leads to the maximal violation of $W_3$. 
Unlike the self-testing result presented in \cite{Sarkar2024networkquantum} here we assume that the sources are identical, that is, the sources generate the same state. Let us also remark here that we consider that the states generated by the sources are pure and the measurements are projective. This is a standard assumption is most Bell or steering based self-testing, however, there might be some exceptions as noted in \cite{baptista2023}. 
Consequently, the certification of the states and measurements are weaker compared to the standard self-testing protocols. %However, the assumption of identical sources in principle can be verified practically by interchanging all the sources among each other and ensuring that they still reproduce the same correlations. 
Let us now state our self-testing result in the triangle network scenario with one trusted party. 
\setcounter{thm}{0}
\begin{thm}\label{Theo1} 
Assume that Alice, Bob and Charlie observe that the correlations $p(a,b,c)$  in the steering functional $W_3$ \eqref{steering-3}  are equal along with $W_3=1$ such that trusted Alice performs the Bell-basis. Considering identical sources, such that the states generated by them are pure and measurements of untrusted parties are projective, the following statements hold true:
\\

1.  \ \  There exist unitary transformation $U$,  such that the state generated by the sources are certified as
\begin{eqnarray}\label{lem1.2}
(\mathbbm{1}_{A_1^1}\otimes U_{A^2_2})\ket{\psi_{A_1^1A_2^2}}=(\mathbbm{1}_{A_1^2}\otimes U_{A^1_3})\ket{\psi_{A_1^2A_3^1}}=\nonumber\\(\mathbbm{1}_{A_2^1}\otimes U_{A^2_3})\ket{\psi_{A_2^1A_3^2}}=\ket{\phi^+}.\ \ 
\end{eqnarray}
\\
2.  \ \  The measurement of Bob $\{M_{A_2}^{a_2}\}$ and Charlie $\{M_{A_3}^{a_3}\}$ is certified as
\begin{eqnarray}\label{lem1.1}
(\I_{A_2^1}\otimes U_{A^2_2})\,M_{A_2}^{a_2}\,(\I_{A_2^1}\otimes U_{A^2_2}^{\dagger})&=&\proj{\phi_{a_2}},\nonumber\\ 
(U_{A^1_3}\otimes U_{A^2_3})\,M_{A_3}^{a_3}\,(U_{A^1_3}^{\dagger}\otimes U_{A^2_3}^{\dagger})&=&\proj{\phi_{a_3}}
\end{eqnarray}
for all $a_2,a_3=0,1,2,3$.
\end{thm}
The proof of the above fact is given in the Appendix. Let us now consider a generalisation of the triangle network [Fig. \ref{fig1}] to a network with arbitrary number of nodes \figref{fig2}. 

{\section{$n$-ring scenario}} As previously done for the triangle, we trust node $A_1$ to measure in the Bell basis \eqref{Amea1}. Generalising the notion of swap-steerability to this network, any correlations $p(\{a\}_n)$ on $n-$ring network (ring network with $n$ nodes) that admits a SOHS model must be of the form
\begin{equation}\label{SOHS-n}
    p(\{a\}_n)= \sum_{\{\lambda\}}\prod_{j=1}^{n} p(\lambda_{j})\Tr{[N^{a_1}\rho_{\lambda_1}\otimes\rho_{\lambda_n}]}\prod_{i=2}^{n} p(a_{i}|\lambda_{i-1},\lambda_{i}).
\end{equation}

Let us now construct a witness to observe swap-steering in the $n-$ring network. For this purpose, consider the network with $n-1$ nodes with the corresponding witness given by $W_{n-1}=\sum_{\{a\}_{n-1}}c_{\{a\}_{n-1}}p(a_1,\ldots,a_{n-1})$. Then, to obtain the witness $W_n$ corresponding to the network with $n$ nodes, in the witness $W_{n-1}$ the outcome of the $n-1$ party $a_{n-1}$ is "reverse coarse-grained" or refined as
\begin{eqnarray}\label{rcg}
  0 &\rightarrow (0,0) , \quad (1,1) ,\quad  (2,2) , \quad (3,3)  \nonumber\\
   1 &\rightarrow (0,1) , \quad (1,0) ,\quad  (2,3) , \quad (3,2)  \nonumber \\
   2 &\rightarrow (0,2) , \quad (2,0) ,\quad  (1,3) , \quad (3,1) \nonumber \\
   3&\rightarrow  (0,3) , \quad (3,0) ,\quad  (1,2) , \quad (2,1) .
\end{eqnarray}

%For this purpose, we recall the labelling of the outcomes introduced below Eq. \eqref{steering-3}: $0 \rightarrow \alpha$, $1 \rightarrow \beta$, $2 \rightarrow -\alpha$ and $3 \rightarrow -\beta$ and then construct the set $\mathcal{P}_{W_{n}} = \{\{a\}_n| \prod_{i=1}^{n}a_i = \alpha^{n-2j}\beta^{2j}\}$, with $j = 0,1,2,\dots$ $n/2$ or $(n-1)/2$ for even or odd $n$ respectively. Now the witness $W_n$ is defined as the sum of the probabilities of the elements of $\mathcal{P}_{W_{n}}$, that is, 
%\begin{eqnarray}\label{witness-n}
 %   W_n = \sum_{\mathcal{P}_{W_n}} p(\{a\}).
%\end{eqnarray}
%
%

Notice that if we consider the witness $W_2$ in \cite{Sarkar2024networkquantum} in the network with two nodes $n=2$, the witness $W_3$ \eqref{steering-3} can be obtained from it using the above "reverse" coarse-graining \eqref{rcg}. Following the above procedure, the witness in the network with $n$ nodes is given by $ W_n=\sum_{\{a\}_n}c_{\{a\}_n} p(\{a\}_n)$ where
\begin{equation}\label{witness-na}
   c_{\{a\}_n}=\begin{cases}
       1\ \mathrm{if}\ a_1-a_2-\sum_{i=3}^n(-1)^{\sum_{k=2}^{i-1}a_k}a_i\ \mathrm{mod}\ 4 = 0\\
       0\ \mathrm{otherwise}.
   \end{cases}
\end{equation}

Let us now find the maximum value of $W_n$ attainable using the SOHS model for the $n-$ring network.

\begin{figure*}[t!]%
    \centering
    {{\includegraphics[width=8.5cm]{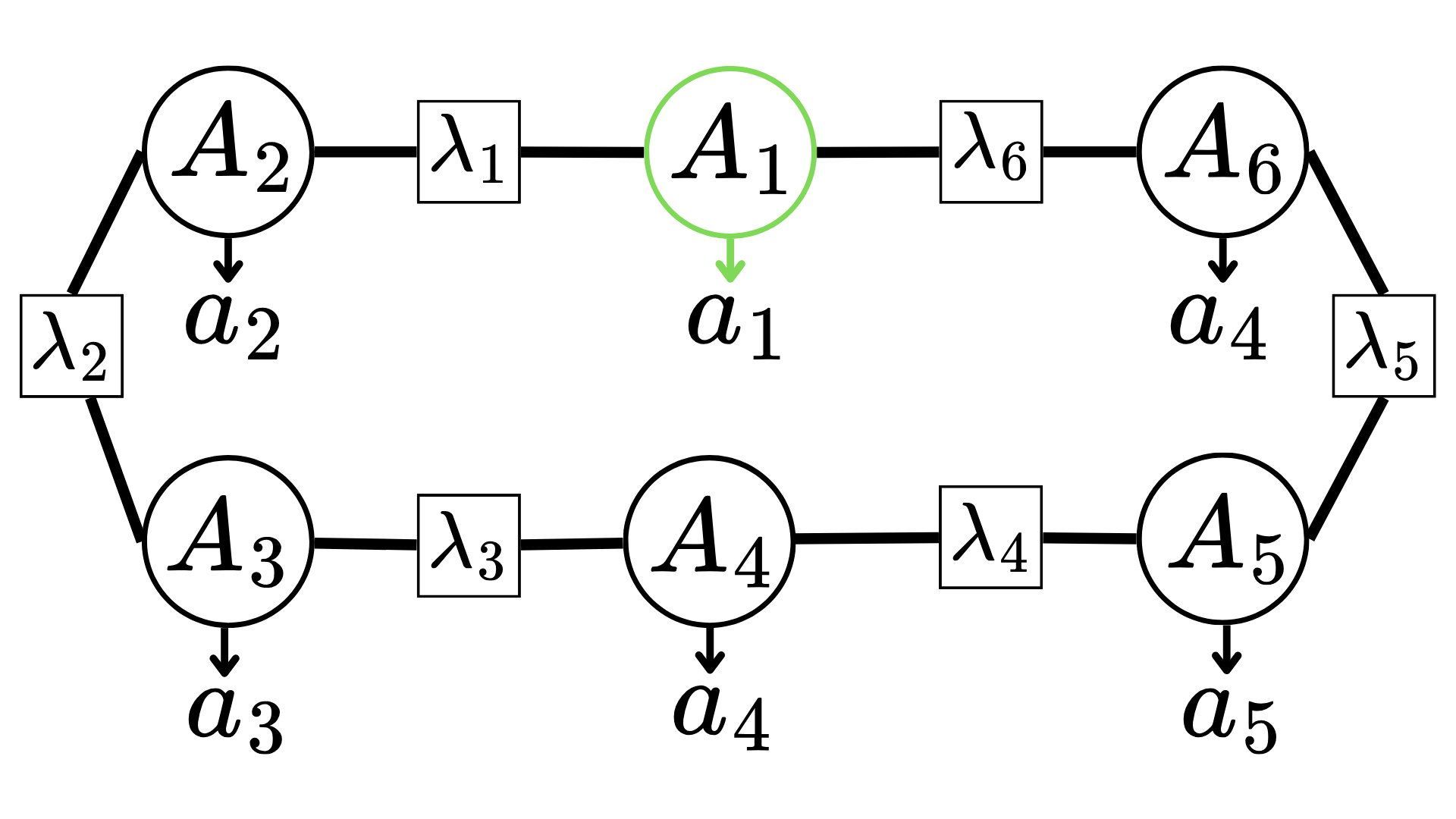} }}%
{{\includegraphics[width=9cm]{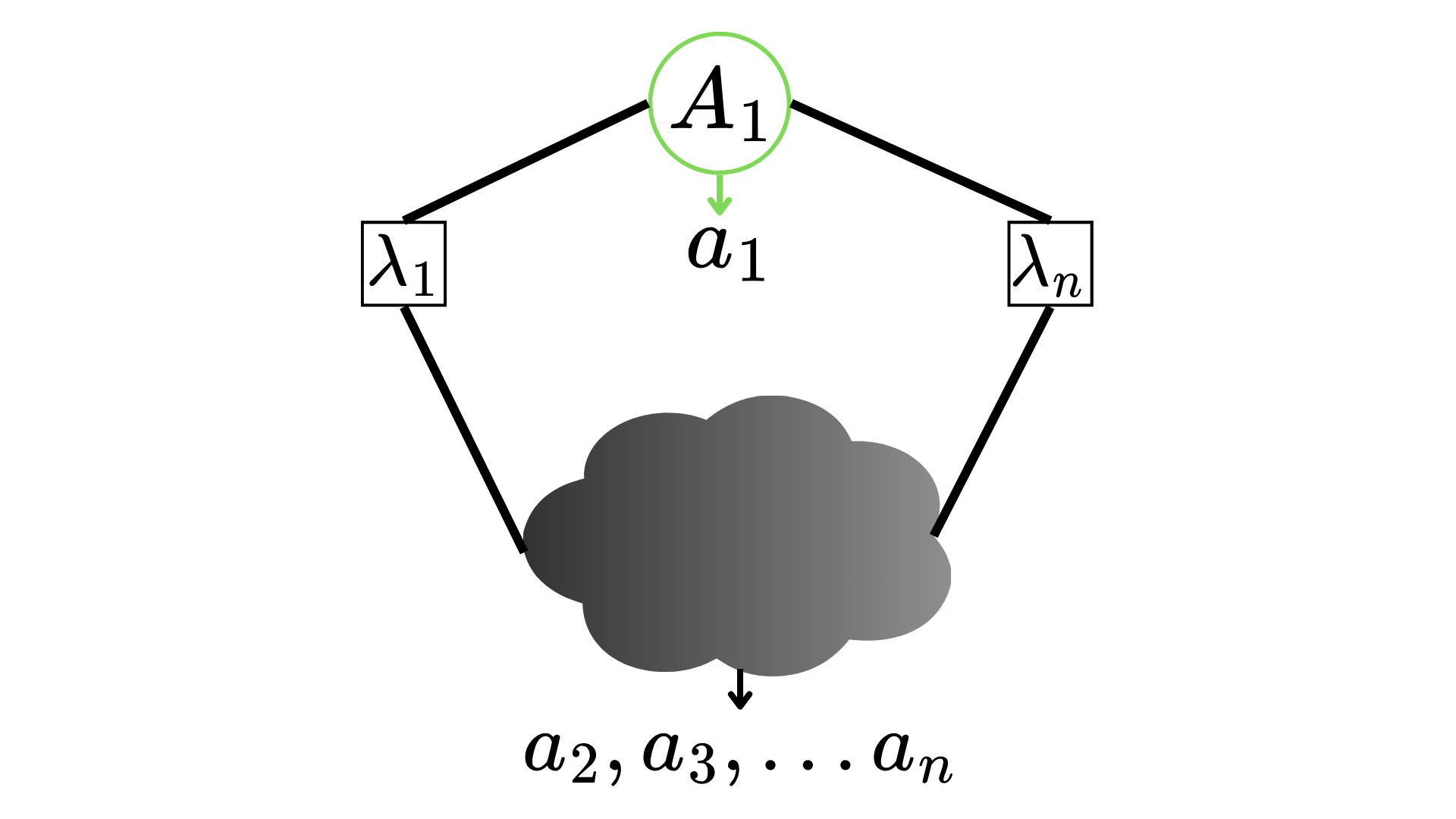} }}%
     \caption{(left) Hexagonal swap-steering scenario. (right) $n-$ring topologically robust swap-steering scenario. In both scenarios, %spatially separated labs $A_i$ receive subsystems from sources $\lambda_{i-1},\lambda_i$ and perform a single four-output measurement.
     the node $A_1$ is trusted to perform measurement in the Bell-basis and all other parties perform a single four outcome measurement. In the left scenario, $A_i$ receive subsystems from sources $\lambda_{i-1},\lambda_i$. However, in the right one all the parties $A_2, \ldots,A_n$ can do any operation among themselves to produce the outputs $(a_2,\ldots,a_n)$. After the experiment, they construct the joint distribution $p(\{a\})$.} 
    \label{fig2}
\end{figure*}

\begin{fakt}\label{fact5}
 The maximum value $\beta_{SOHS}$ of the three party swap-steering functional $W_n$ that can be obtained using a SOHS model is $\frac{1}{2}$.
\end{fakt}
The proof of the above is stated in Fact 1 of the appendix where we find the upper bound for more general correlations admitting a topologically robust SOHS model stated below \eqref{SOHS-nc}. %For the detailed proof refer to Appendix.
\iffalse
\begin{proof}
    Refer to supplementary material for detailed proof. Here we will only state the arguments of the proof. Firstly, observe that, under our labeling scheme, for any $\{a\}_n$, $\prod_{i=1}^{n}a_i = \alpha^{n-j}\beta^{j}$. Now, we evaluate the witness \eqref{witness-n}.
    \begin{eqnarray}
        \sum_{\mathcal{P}_{W_n}} \sum_{\{\lambda\}} p(\{\lambda\}) \Tr{[N^{a_1}_{A_{1}}\rho_{\lambda_1}\otimes\rho_{\lambda_n}]}\prod_{m}^{n-1} p(a_{m}|\lambda_{m-1},\lambda_{m}) \nonumber \\ \leq \frac{1}{2}\sum_{\mathcal{P}_{W_n}}\sum_{\{\lambda\}} p(\{\lambda\}) \prod_{m}^{n-1} p(a_{m}|\lambda_{m-1},\lambda_{m}) \qquad\quad
    \end{eqnarray}
    Notice that $a_1$ is pulled out of the summation. Also, observe that upon looking at the elements of the set $\mathcal{P}_{W_{n}}$ and ignoring the values of $a_{1}$, we obtain the set of all possible $n-1$ strings $\{a\}_{n-1}$. Obviously, $\sum_{\{a\}_{n-1}} p(\{a\}_{n-1}|\{\lambda\})= 1$ and $\sum_{\{\lambda\}} p(\{\lambda\}) = 1$. Using this, we obtain
    \begin{eqnarray}
        W_n \leq \beta_{SOHS} = \frac{1}{2}
    \end{eqnarray}
   This concludes the proof.
\end{proof} 
\fi
Consider now the following quantum experiment. Each source prepares the state $\ket{\psi_i}=\ket{\phi_+}_{A_i^2A_{i+1}^1}$ and each node measures in the same basis as the trusted node, that is, $M_{A_{n}}=\{\proj{\phi_{+}},\proj{\phi_{-}},\proj{\psi_{+}},\proj{\psi_{-}}\}_{A_{n}^1,A_n^2}$ where the corresponding states are given in \eqref{Amea1}. %Using these states, measurements and output labels $\ket{\phi_{\pm}} \rightarrow \pm\alpha, \ket{\psi_{\pm}} \rightarrow \pm\beta$,
By assuming that $W_{n-1} = 1$ using the above quantum strategy, one can prove by induction that $W_n = 1$ using the same strategy. This is understood as the Bell basis measurement at the $n-$th node simply swapping entanglement to the $(n-1)-$th and 1-st node in a manner consistent with the coarse-graining scheme \eqref{rcg}. These states will in fact correspond to a Bell basis and thus it reduces to a network with $n-1$ nodes with the above-stated quantum states and measurements. This is also proven in Fact 2 of the appendix.
Notice that this is also the algebraic bound of $W_n$. Thus, the experiment described above is an example of the maximal quantum violation of the steering inequality. Similar to the triangle network,  all the sources must necessarily distribute entangled states and all the nodes must measure an entangled basis to generate $n-$ring swap-steering correlations.

%\begin{proof}
 %   Detailed proof can be found in the supplementary material. The arguments presented for the case of triangle networks can be applied directly to construct the proof of the above statement. If any of the measurements or states are separable, there is a break in the chain of entanglement swapping across the network and the state at the trusted node is separable. Thus, they can not generate steering correlations.
%\end{proof}

{\subsection{Topologically Robust Swap Steering}}
Let us now consider the scenario where the $n-1$ untrusted parties are conncected via addtional sources in an arbitrary way. The strongest such scenario would be one where the untrusted parties and the sources connecting them are morphed into one single party that receives the hidden variables $\lambda_1$ and $\lambda_n$, and outputs a distribution of $n-1$ outcomes [see Fig. \ref{fig2}]. % and can even communicate among each other. %We still demand that this party produce a $n-1$ digit 4-bit string. This is equivalent to the scenario where each untrusted party is connected to every other untrusted party. 
Here we will prove that the SOHS bound for this scenario will remain the same. Note that we still require that the trusted party receives a separable two-qubit state. This will also serve as the bound for scenarios in which the untrusted parties posses weaker connections among themselves, for example the ring scenario considered previously. The SOHS models for this scenario will be termed as topologically robust SOHS (TSOHS) models and will be given by 
\begin{eqnarray}\label{SOHS-nc}
  p(\{a\}_n)= \sum_{\lambda_1,\lambda_n} p(\lambda_{1},\lambda_n)\Tr{[N^{a_1}\rho_{\lambda_1}\otimes\rho_{\lambda_n}]} \nonumber\\p(\{a\}_{n-1/1}|\lambda_1,\lambda_n).
\end{eqnarray}
Formally, we arrive at the above expression of TSOHS models by modifying \eqref{SOHS-n}, imposing weaker conditions than outcome-independence \eqref{ass1} and separable sources \eqref{ass2} given by:
\setcounter{thm}{2}
\begin{assu}[Outcome-independence of the trusted party]\label{ass1} The measurement outcome of the trusted party is independent of any other outcome if one has access to the hidden variables $\lambda_i$ from all the sources they are connected to.
\begin{eqnarray}\label{oi-eq2}
p(a_1|\{\lambda\},\{a\}_{n-1/1})=p(a_1|\lambda_1,\lambda_n)
\end{eqnarray}
or equivalently,
\begin{equation}\label{oi-eq3}
p(\{a\}_{n-1/1}|\{\lambda\},a_1)=p(\{a\}_{n-1/1}|\lambda_1,\lambda_n)\ \ 
\end{equation}
\end{assu}
\begin{assu}[Separable state with trusted party]\label{ass2} The joint state sent by the source $S_1,S_n$ to the trusted party is separable, that is, $\rho_{\lambda_1,\lambda_n}=\sum_{k}p^{(k)}\rho_{\lambda_1}^{(k)}\otimes\rho_{\lambda_n}^{(k)}$.
\end{assu}
As we consider linear witnesses, without loss of generality, we consider the joint state received by the trusted party to be a product state.
%Within quantum theory, this would imply that if the sources generate the states $\rho_i$, then the joint state generated by the sources $\rho_{12}$ will be given by $\rho_{12}=\rho_{1}\otimes\rho_2$.Unlike \eqnref{SOHS-n}, the distribution of untrusted outcomes is now conditioned on both sources and is no longer in product form. 
It is straightforward to observe that even for this scenario by considering the witness $W_n$, we have $W_n\leq\beta_{TSOHS}$ such that $\beta_{TSOHS}=\frac{1}{2}$. The proof is specified in Fact 1 of the Appendix.

%\begin{fakt}\label{fact5}
 %    The maximum value $\beta_{SOHS}$ that can be achieved using correlations that admit a SOHS model \eqref{SOHS-nc} of $W_n$ is $\frac{1}{2}$.
%\end{fakt}
%The proof for this remains the same as for the $n-$ring witness and can be found in the appendix. 
Consequently, the observed swap-steering inequality from the violation of \eqref{SOHS-nc} is robust against any network structure among the untrusted parties. This would also include cases when the untrusted parties might communicate with each other. Consequently, one can detect a quantum advantage in a network against much more powerful classical models, which is impossible in the standard multipartite quantum steering scenarios. Also, notice from Fact \ref{fact5} that the gap between the quantum and SOHS values is considerably large and thus it can be observed in experiments. Consequently, this is the first example of a topologically noise robust advantage in a quantum network. 

Notice that with no information about the network structure, observing correlations compatible with $n$ party steering does not imply that the correlation was generated genuinely by $n$ different parties. This is simply an inverse of the topological robustness property of swap-steering. For example, consider the experiment on the minimal scenario \cite{Sarkar2024networkquantum} with a single untrusted party generating maximal swap-steering violation. The untrusted party can simply use the refinement scheme [Eq. \eqref{rcg}] used to extend the witness to the $n$ party scenario to his output and report a $n-1$ string bit. This by definition, will result in statistics that maximally violate the $n$ party witness $W_n$. Thus, one cannot naively conclude sucessful $n$ party steering by only observing the correct statistics, it must be supplemented with some information about the network structure. Note that this is a core feature of topological robustness which allows us to observe steering without any information about the network structure among untrusted parties. % and it is logically equivalent to the arguments in this section in a reversed order, where instead of compressing all untrusted parties into a single more powerful party we expand a single untrusted party into multiple.}

%Unlike the previous cases, since here the network structure among the untrusted parties is not assumed, only the two sources connected to the trusted party [see Fig. \ref{fig2}] need to generate entangled states to observe topologically robust network steering.

{\subsection{All bipartite entangled states generate swap-steering correlations in the $n-$ring network}}
Let us finally show that every bipartite entangled state generates swap-steerable correlations in the $n-$ring network. The following construction is inspired from \cite{Branciard_2013, sarkar2024_every}.
Consider again the $n-$ring network [see Fig. \ref{fig2}], such that the trusted party, instead of performing a single measurement, has access to a tomographically complete set of measurements. Furthermore, we allow the trusted party to receive quantum states of arbitrary dimensions from both the sources connected to it. Thus, the trusted node $A_1$ has access to $d_{A_{1}^1}^2d_{A_{1}^2}^2$ measurements, where $A_{1}^1, A_{1}^2$ represent the two subsystems received by the trusted node of dimension $d_{A_{1}^1},d_{A_{1}^2}$ respectively. The measurements of the trusted party are denoted as $\mathcal{A}_{st}=\{\tau_s\otimes\omega_t,\I-\tau_s\otimes\omega_t\}$ where $\tau_s\in \mathbb{C}^{d_{A_1^1}}, \omega_t\in \mathbb{C}^{d_{A_1^2}}$ are projectors such that $\{\tau_s\}_s,\{ \omega_t\}_t$ span the Hilbert spaces $\mathbb{C}^{d_{A_1^1}}, \mathbb{C}^{d_{A_1^2}}$ respectively and thus are tomographically complete. Consequently, we have that $s=\{1,\ldots,d_{A_1^1}^2\}$ and $t=\{1,\ldots,d_{A_1^2}^2\}$. 
All the other parties perform a single measurement. Moreover, every measurement now results in binary outcomes denoted as $a_i=0,1$ for $i=1,\ldots,n$.

Consider now an entangled state $\tilde{\rho}\in\mathbb{C}^{d_{A_1^1}}\otimes\mathbb{C}^{d_{A_1^2}}$ such that the corresponding entanglement witness is given by $W_{\tilde{\rho}}$, that is, $\Tr(W_{\tilde{\rho}}\tilde{\rho})<0$ and  $\Tr(W_{\tilde{\rho}}\sigma)\geq0$ for any seperable state $\sigma$. Now, any operator $W_{\tilde{\rho}}$ can be expressed using the tomographically complete set of operators stated above as $W_{\tilde{\rho}}=\sum_{s,t}\beta_{\tilde{\rho},s,t}\tau_s\otimes\omega_t$. Utilising this entanglement witness, let us now propose the following $n-$ring swap-steering inequalities 
\begin{eqnarray}\label{Wituniv}
    \mathcal{S}_{\tilde{\rho}}=-\sum_{s,t}\beta_{\tilde{\rho},s,t}p(0,0,\ldots,0|st)
\end{eqnarray}
where $p(0,0,\ldots,0|st)$ is the probability of obtaining outcome $0$ by all the parties given the input $s,t$ of the trusted party $A_1$.
Let us now find its SOHS bound.
\begin{fakt}
    Consider the $n-$ring swap-steering scenario  and the functional $\mathcal{S}_{\tilde{\rho}}$ \eqref{Wituniv}. The maximal value attainable of $\mathcal{S}_{\tilde{\rho}}$ using an SOHS model is $\beta_{SOHS}=0$. 
\end{fakt}

The proof of the above fact is in the Appendix.
Let us now consider that the source connecting the parties $A_2,A_3$ generates the state $\tilde{\rho}_{A_2^1,A_3^2}\in\mathbb{C}^{{d_{A_1^1}}}\otimes\mathbb{C}^{d_{A_1^2}}$. Now the source connecting the trusted party $A_1$ with $A_2$ produces the maximally entangled state of dimension $d_{A_1^1}$, that is,  $\ket{\psi_{A_1^1,A_2^2}}=\ket{\phi^+_{d_{A_1^1}}}$. Rest of the other sources generate the maximally entangled state $\ket{\phi^+_{d_{A_1^2}}}$.  Moreover, $A_2$ performs the measurements $\{\proj{\phi^+}_{d_{A_1^1}},\I-\proj{\phi^+}_{d_{A_1^1}}\}, $ and rest of the parties $A_3,\ldots,A_n$ perform $\{\proj{\phi^+}_{d_{A_1^2}},\I-\proj{\phi^+}_{d_{A_1^2}}\}$ respectively. Now, consider a party $A_i$ for $i=4,\ldots,n$ obtains outcome $0$, then the post-measurement state between $A_{i-1},A_{i+1} (n+1\equiv 1)$ is the maximally state of dimension $d_{A_{1}^2}$. Consequently, it is straightforward to observe that the $n-$network reduces to a triangle network such that $A_1,A_3$ share $\ket{\phi^+_{d_{A_1^2}}}$. Now, $A_2,A_3$ projecting the joint state over their $0-$th outcome swaps the state $\tilde{\rho}$ to the trusted $A_1$ and thus one is left with the expression
\begin{eqnarray}
    \mathcal{S}_{\tilde{\rho}}&=&-\frac{1}{d_{A_1^1}^2d_{A_1^2}^{n-1}}\sum_{s,t}\beta_{\tilde{\rho},s,t}\Tr(\tau_{s}\otimes\omega_{t}\tilde{\rho}_{A_1^1A_1^2})\nonumber\\&=&-\frac{1}{d_{A_1^1}^2d_{A_1^2}^{n-1}}\Tr(W_{\tilde{\rho}}\ \tilde{\rho})>0.
\end{eqnarray}
Notice that the factor $1/{d_{A_1^1}^2d_{A_1^2}^{n-1}}$ takes into account the probability of obtaining outcome $0$ by all the parties. 
%evaluating $$ gives us
%\begin{eqnarray}
 %   \mathcal{S}_{\tilde{\rho}}&=&-\sum_{s,t}\beta_{\tilde{\rho},s,t}\Tr(\Big[\tau_{s,A_0}\otimes\omega_{t,A_1}\otimes\proj{\phi^+_{B_0A_2^1}}\otimes\proj{\phi^+_{C_0A_3^1}}\Big]\proj{\phi^+_{A_0A_2^1}}\otimes\proj{\phi^+_{A_1A_3^1}}\otimes\tilde{\rho}_{B_0C_0})\nonumber\\
  %  &=&-\frac{1}{d_{A_0}d_{A_1}}\sum_{s,t}\beta_{\tilde{\rho},s,t}\Tr(\Big[\proj{\phi^+_{B_0A_2^1}}\otimes\proj{\phi^+_{C_0A_3^1}}\Big]\tau_{s,A_2^1}\otimes\omega_{t,A_3^1}\otimes\tilde{\rho}_{B_0C_0})\nonumber\\
   % &=&-\frac{1}{d_{A_0}^2d_{A_1}^2}\sum_{s,t}\beta_{\tilde{\rho},s,t}\Tr(\tau_{s,B_0}\otimes\omega_{t,C_0}\tilde{\rho}_{B_0C_0})=-\frac{1}{d_{A_0}^2d_{A_1}^2}\Tr(W_{\tilde{\rho}}\ \tilde{\rho})>0.
%\end{eqnarray}
Consequently, every entangled state results in $n-$ring swap-steerable correlations. It is important to note here that although the above prescription closely follows \cite{Branciard_2013}, the sources connected to the trusted party are also trusted in \cite{Branciard_2013} which is not the case here.

{\section{Discussions}}
In this work we extend the results of \cite{Sarkar2024networkquantum} by demonstrating swap-steering on the triangle network with one trusted party. Additionally, we show that swap-steering can be detected only if all states generated by the sources and measurements at the nodes are entangled. We also demonstrate a weak self-testing scheme for states and measurements on the triangle network.  In the appendix, we provide a noise tolerance analysis of the swap-steering witness in the triangle network. In particular, we consider the case when the state and measurement that generates the maximal value of the witness are mixed with noise and then find the amount of noise that would erase the quantum violation. Then, we extend the construction of the witness and the quantum experiment, which maximally violates the SOHS bound, to the $n-$ring network. We prove that our witness can distinguish between quantum and SOHS models, even if we cannot trust the network structure between the untrusted parties with a large noise tolerance.  While the nonlocal distributions in \cite{boreiri2024topologicallyrobustquantumnetwork} were also recently shown to be topologically robust, it is very sensitive to noise \cite{boreiri2024noiserobustproofsquantumnetwork,Kriv_chy_2020}. The combined topological and noise robust character makes this swap-steering scenario an excellent candidate for experimental implementations of quantum network advantage. %Contrary to this, swap-steering requires that the trusted party perform a single fixed measurement. Also, the construction of a witness. This also makes our scheme experimentally friendly as one has to consider less number of correlations to witness quantum steering in networks. However, the measurement elements of the trusted party are maximally entangled. 
Finally, we demonstrated that if the trusted party has access to complete tomography then one can detect swap-steering of all bipartite entangled states. 

There are a few interesting open problems that follow up from this work. %We notice that by moving from a non-locality scenario to a steering scenario, one can work with a linear witness. 
The additional assumption of trusted measurement drastically eases the analysis, allowing one to witness steering using linear functionals
An interesting question would be to study the space of quantum correlations that could be obtained in the network with trusted measurement. %We might expect convex facets in the steering set, although, a linear witness is not sufficient to conclude this. \\ 
Network non-locality relies on the assumptions of independent sources and a fixed network structure. However, these assumptions require control over the environment outside a trusted lab, which is undesirable. Topological robustness allows us to drop these assumptions, which might enable secure communication tasks on networks. It will also be interesting to observe how much trust is needed in one of the parties to observe swap-steering similar to \cite{Sarkar_distrust}.
%\newpage
%\begin{acknowledgments}
 %We would like to thank Stefano Pironio for reviewing the manuscript and providing critical comments that considerably improved the manuscript. 

 {\section{Acknowledgements}}We thank Stefano Pironio for mentoring this work and providing crucial insights. 
 This project was funded within the QuantERA II Programme (VERIqTAS project) that has received funding from the European Union’s Horizon 2020 research and innovation programme under Grant Agreement No 101017733. S.S. also acknowledges National Science Centre, Poland, grant Opus 25, 2023/49/B/ST2/02468.

\providecommand{\noopsort}[1]{}\providecommand{\singleletter}[1]{#1}%

 \onecolumngrid
\appendix

\setcounter{fakt}{0}
\section{Proofs in the Triangle scenario}

%Recalling the labelling scheme: $0 \rightarrow \alpha$, $1 \rightarrow \beta$, $2 \rightarrow -\alpha$ and $3 \rightarrow -\beta$, 
The swap-steering witness in the triangle scenario is given by $W_3 = \sum_{a_1,a_2,a_3=0}^3 c_{a_1,a_2,a_3}p(a_1,a_2,a_3)$
\begin{equation}
   c_{a_1,a_2,a_3}=\begin{cases}
       1\ \mathrm{if}\ a_1-a_2-(-1)^{a_2}a_3\ \mathrm{mod}\ 4 = 0\\
       0\ \mathrm{otherwise}.
   \end{cases}
\end{equation}

\subsection{Noise tolerance analysis}
Let us now evaluate the effect of introducing noise in the experiment that maximally violates $\beta_{SOHS}$. As of now, we consider the simplest noise model, that is, all the states are mixed with white noise depending on the parameter $v_i$ and all measurement elements are mixed with white noise parametrized with $w_j$.
Consequently, we model the noisy source as distributing the Werner state given by
\begin{eqnarray}\label{Werner}
    \rho(v_i)_{A_i^2A_{i+1}^1}=v_i\proj{\phi_{+}}_{A_i^2A_{i+1}^1}+(1-v_i)\frac{\I}{4}\qquad i=1,2,\ldots,n.
\end{eqnarray}
Similarly, we will add local measurement noise using the following white noise model
\begin{eqnarray}\label{Werner}
    \tilde{M}^{a_i}_{A_i}(w_i)=w_iN^{a_i}_{A_i}+(1-w_i)\frac{\I}{4} \qquad i=2,\ldots,n.
\end{eqnarray}
Here, $v,w$ are the state and measurement visibility respectively and $N^{a_i}$ refers to all of the Bell basis projectors. As in our scenario, we trust $A_1$ to perform a perfect Bell basis measurement, so we only introduce noisy measurement in the other two nodes. If each source $S_i$ distributes the state $\rho_{A_i^2A_{i+1}^1}(v_i)$ and the other two nodes measure in the basis $M_{A_i}(w_i)$, we evaluate $W_3$ to obtain the following expression.
\begin{eqnarray}
    W_3 = \frac{1}{4}(1+3v_1v_2v_3w_2w_3) = \frac{1}{4}(1+3\Tilde{v}).
\end{eqnarray}
For steering correlations $W_3 > 1/2$, thus $\Tilde{v}>1/3$. Here $\Tilde{v}$ acts as the combined visibility obtained by multiplying all the visibilities. So, to set up an experiment to observe triangle swap-steering, one must arrange for sources and measurements such that $\Tilde{v}>1/3$.

\subsection{Triangle network weak self-testing}

Before proceeding let us remark here that for ease of notation in the proof we represent $A^{i}_1\equiv A_i, A^{i}_2\equiv B_i, A^{i}_3\equiv C_i$ for $i=1,2$. Moreover the parties $A_1,A_2,A_3$ are called Alice, Bob and Charlie repsectively.
\setcounter{thm}{0}
\begin{thm}\label{Theo1} 
Assume that Alice, Bob and Charlie observe that the correlations $p(a,b,c)$  in the steering functional $W_3$ are equal [see Eq. \eqref{witness-na1} below for $n=3$] along with $W_3=1$ such that trusted Alice performs the Bell-basis. Considering identical sources, such that the states generated by them are pure and measurements of untrusted parties are projective, the following statements hold true:
\\
\\
1.  \ \  There exist unitary transformation $U$,  such that the state generated by the sources are certified as
\begin{eqnarray}\label{lem1.2}
(\mathbbm{1}_{A_1}\otimes U_{B_2})\ket{\psi_{A_1B_2}}=(\mathbbm{1}_{A_2}\otimes U_{C_1})\ket{\psi_{A_2C_1}}=(\mathbbm{1}_{B_1}\otimes U_{C_2})\ket{\psi_{B_1C_2}}=\ket{\phi^+}.\ \ 
\end{eqnarray}
\\
\\
2.  \ \  The measurement of Bob $\{M_b\}$ and Charlie $\{M_c\}$ is certified as
\begin{eqnarray}\label{lem1.1}
[\I_{B_1}\otimes U_{B_2}]\,M_b\,[\I_{B_1}\otimes U_{B_2}^{\dagger}]=\proj{\phi_b},\quad [U_{C_1}\otimes U_{C_2}]\,M_c\,[U_{C_1}^{\dagger}\otimes U_{C_2}^{\dagger}]=\proj{\phi_c}.
\end{eqnarray}
\end{thm}
\begin{proof}
Considering now that the states generated by the identical sources are pure and measurements are projective. As Alice's measurement is trusted and known to perform the Bell-basis which acts on $\mathbb{C}^2\otimes\mathbb{C}^2$, the states $\ket{\psi}$ shared between Alice and Bob can be written using Schmidt decomposition as 
\begin{eqnarray}\label{state1}
  \ket{\psi_{1}}= \ket{\psi_{A_1B_2}}=\sum_{i=0,1}\lambda_{i}\ket{e_{i}}_{A_1}\ket{f_{i}}_{B_2} 
\end{eqnarray}
where $\lambda_{i,1}\geq0$ and $\{\ket{e_{i,1}}\},\{\ket{f_{i,1}}\}$ form an orthonormal basis for each $i$. As the sources are identical, the state shared among Alice-Charlie and Bob-Charlie is the same as \eqref{state1} given as
\begin{eqnarray}
    \ket{\psi_{2}}= \ket{\psi_{B_1C_2}}=\sum_{i=0,1}\lambda_{i}\ket{e_{i}}_{B_1}\ket{f_{i}}_{C_2} 
\end{eqnarray}
and,
\begin{eqnarray}
\ket{\psi_{3}}=\ket{\psi_{A_2C_1}}=\sum_{i=0,1}\lambda_{i}\ket{e_{i}}_{A_2}\ket{f_{i}}_{C_1}.
\end{eqnarray}
Now applying a unitary $U$ on these states such that $U\ket{f_{i}}=\ket{e^*_{i}}$ gives us for any $s=1,2,3$
\begin{eqnarray}
\ket{\tilde{\psi}_s}=\I\otimes U\ket{\psi}=\sum_{j=0,1}\lambda_{i}\ket{e_{i}}\ket{e^*_{i}}.
\end{eqnarray}
Now, notice that the state on the right-hand side can be represented as
\begin{eqnarray}\label{state1}
\ket{\tilde{\psi}_s}= \I\otimes P\ket{\phi^+}=P^{*}\otimes\I\ket{\phi^+}
\end{eqnarray}
where
\begin{eqnarray}\label{P}
    P=\sqrt{2}\sum_{j=0,1}\lambda_{i}\proj{e_{i}^*}.
\end{eqnarray}

Now, considering all $p(0,b,c)$ for all $b,c=0,1,2,3$ such that $p(0,b,b)=1/16$ and $p(0,b,c)=0$ for $b\ne c$,  we have that
\begin{eqnarray}
    p(0,b,c)=\bra{\psi_1}\bra{\psi_2}\bra{\psi_3}\left[\proj{\phi_0}_{A_1A_2}\otimes M_{b,B_1B_2}\otimes N_{c,C_1C_2}\right]\ket{\psi_1}\ket{\psi_2}\ket{\psi_3}.
\end{eqnarray}
Now, using \eqref{state1}, we have from the above formula
\begin{eqnarray}\label{43}
    p(0,b,c)=\bra{\phi^+_{A_1B_2}}\bra{\phi^+_{B_1C_2}}\bra{\phi^+_{A_2C_1}}\left[\proj{\phi_0}_{A_1A_2}\otimes \tilde{M}_{b,B_1B_2}\otimes \tilde{N}_{c,C_1C_2}\right]\ket{\phi^+_{A_1B_2}}\ket{\phi^+_{B_1C_2}}\ket{\phi^+_{A_2C_1}}
\end{eqnarray}
where $\tilde{M}_{b}=P^*_{B_1}\otimes P_{B_2}[ U_{B_2} M_b U_{B_2}^{\dagger}] P^*_{B_1}\otimes P_{B_1},\ \tilde{N}_{c}=P_{C_1}\left[U_{C_1}\otimes U_{C_2} N_cU_{C_1}^{\dagger}\otimes U_{C_2}^{\dagger}\right]P_{C_1}$ where $P$ is given in Eq. \eqref{P}. Let us now notice that $\ket{\phi^+_{A_1B_2}}\ket{\phi^+_{A_2C_1}}=\ket{\phi^+_4}_{A_1A_2|C_1B_2}$ where $\ket{\phi^+_{4}}$ is the maximally entangled state of local dimension $4$. Furthermore, $\Tr_{A_1A_2}[\proj{\phi_a}_{A_1A_2}\otimes\I_{C_1B_2}\proj{\phi^+_4}]=\frac{1}{4}\proj{\phi_a}_{C_1B_2}.$ Consequently, we have from the above formula \eqref{43} that
\begin{eqnarray}
\bra{\phi^+_{B_1C_2}}\bra{\phi^+_{C_1B_2}}\left[\tilde{M}_{b,B_1B_2}\otimes \tilde{N}_{c,C_1C_2}\right]\ket{\phi^+_{B_1C_2}}\ket{\phi^+_{C_1B_2}}=4p(0,b,c).
\end{eqnarray}
Again using the fact that $\ket{\phi^+_{C_1B_2}}\ket{\phi^+_{C_2B_1}}=\ket{\phi^+_4}_{B_1B_2|C_1C_2}$ and thus we have from the above formula
\begin{eqnarray}\label{29}
\bra{\phi^+_{4}}\tilde{M}_{b,B_1B_2}\otimes \tilde{N}_{c,C_1C_2}\ket{\phi^+_{4}}=4p(0,b,c).
\end{eqnarray}
Now, summing the above expression over $b,c$ and using the fact that $\sum_b\tilde{M}_b=P_{B_1}^{*2}\otimes P_{B_2}^2$ and $\sum_c\tilde{M}_c=P_{C_1}^2$ along with $\sum_{b,c}p(0,b,c)=1/4$, gives us
\begin{eqnarray}
    \bra{\phi^+_{4}}P_{B_1}^{*2}\otimes P_{C_2}^2\otimes P_{C_1}^2\otimes \I_{C_2}\ket{\phi^+_{4}}=1.
\end{eqnarray}
Using the identity $\bra{\phi^+_d}R\otimes Q\ket{\phi^+_d}=\frac{1}{d}\Tr (QR^T)$ for any two operators $R,Q$, we obtain from the above expression
\begin{eqnarray}
    \Tr(P_{B_1}^2 P_{C_1}^2)\Tr( P_{C_2}^{*2})=4.
\end{eqnarray}
Using the fact that $\Tr( P_{C_2}^{*2})=2$ and then exploiting the form of $P$ \eqref{P}, we obtain the following condition
\begin{eqnarray}
    2\sum_{i=0,1}\lambda_{i}^4=1
\end{eqnarray}
which using the fact that $\sum_{i=0,1}\lambda_{i}^2=1$ gives us $\lambda_0=\lambda_1=\frac{1}{\sqrt{2}}$.
Consequently, we have that $P=\I$ and thus the states shared among the parties are maximally entangled up to local unitary transformation. 

%Now, we notice that at least one of the elements of Bob's and Charlie's measurements must be entangled. For this purpose, let us consider again the triangle scenario and observe that if Bob performs a separable measurement, then the post-measurement state between Alice and Charlie is separable. Then, for any measurement of Charlie, the resulting state at Alice is separable and thus we can not violate the steering inequality. Consequently, both the measurements of Bob and Charlie need to have at least one entanglement element. 
Let us now certify the measurements. For this purpose, we again consider Eq. \eqref{29} with the fact that $P=\I$ and express it in the observable form as
\begin{eqnarray}\label{50}
    4\sum_bp(0,b,b)=\frac{1}{4}\sum_{x=0}^3\langle\phi^+_4| B_0^x\otimes C_0^{4-x}|\phi^+_4\rangle=1
\end{eqnarray}
where $B_0=\sum_{b=0}^3i^b  U_{B_2} M_b U_{B_2}^{\dagger},$ $C_0=\sum_{c=0}^3i^c U_{C_1}\otimes U_{C_2} N_cU_{C_1}^{\dagger}\otimes U_{C_2}^{\dagger}$. Notice that in expanding the expectation value, one obtains more terms than the right-hand side of the above expression \eqref{50}, however, we use the fact that $p(0,b,c)=0$ for $b\ne c$. 
As $B_0,C_0$ are unitaries, we have that
\begin{eqnarray}
   \langle\phi^+_4| B_0^x\otimes C_0^{4-x}|\phi^+_4\rangle=1\qquad \forall x
\end{eqnarray}
which allows us to conclude that
\begin{eqnarray}
    B_0^x\otimes C_0^{4-x}|\phi^+_4\rangle=|\phi^+_4\rangle
\end{eqnarray}
which allows us to obtain for $x=1$ that $  C_0= B_0.$ Considering now the probabilities $p(a,b,c)$ for $a=1$, we have that
\begin{eqnarray}
    p(1,0,1)+p(1,1,0)+p(1,2,3)+p(1,3,2)=\frac{1}{16}\sum_{x=0}^3(-i)^x\bra{\psi^+_{B_2C_1}}\bra{\phi^+_{B_1C_2}}B_0^x\otimes C_0^x\ket{\psi^+_{B_2C_1}}\ket{\phi^+_{B_1C_2}}=\frac{1}{4}.
\end{eqnarray}
Since each of the expectation values in the above expression is upper bounded by $1$, we have that each term must be $1$ and thus we get for $x=1$
\begin{eqnarray}
\bra{\psi^+_{B_2C_1}}\bra{\phi^+_{B_1C_2}}B_0\otimes (-iC_0)\ket{\psi^+_{B_2C_1}}\ket{\phi^+_{B_1C_2}}=1.
\end{eqnarray}
Notice that $\ket{\psi^+_{B_2C_1}}=\I_{B_2}\otimes \sigma_{x,C_1}\ket{\phi^+_{B_2C_1}}$ where $\sigma_x=\ket{0}\!\bra{1}+\ket{1}\!\bra{0}$ which allows us to conclude from the above formula that
\begin{eqnarray}
\bra{\phi^+_{B_2C_1}}\bra{\phi^+_{B_1C_2}}B_0\otimes \sigma_{x,C_1} (-iC_0)\sigma_{x,C_1}\ket{\phi^+_{B_2C_1}}\ket{\phi^+_{B_1C_2}}=1.
\end{eqnarray}
As all the operators inside the $<.>$ are unitary, we have that
\begin{eqnarray}
    B_0\otimes \sigma_{x,C_1} (iC_0)\sigma_{x,C_1}\ket{\phi^+_4}_{B_1B_2|C_1C_2}=\ket{\phi^+_4}_{B_1B_2|C_1C_2}
\end{eqnarray}
which using the fact that $B_0=C_0$ allows us to conclude that
\begin{eqnarray}
    B_0= (\I\otimes \sigma_x) (-iB_0)^{\dagger}(\I\otimes \sigma_x).
\end{eqnarray}
Consider now an eigenvector $\ket{\xi}$ such that $B_0\ket{\xi}=\ket{\xi}$ which gives us from the above formula
\begin{eqnarray}
    -i(\I\otimes \sigma_x)\ket{\xi}= B_0^{\dagger} (\I\otimes \sigma_x)\ket{\xi}.
\end{eqnarray}
Consequently, $(\I\otimes \sigma_x)\ket{\xi}$ is also an eigenvector of $B_0$ with eigenvalue $i$. 

Similarly, considering all the other probabilities $p(a,b,c)$ for $a=2$, we have that
\begin{eqnarray}
    p(2,0,2)+p(2,2,0)+p(2,1,3)+p(2,3,1)=\frac{1}{16}\sum_{x=0}^3(-1)^x\bra{\phi^-_{B_2C_1}}\bra{\phi^+_{B_1C_2}}B_0^x\otimes C_0^{4-x}\ket{\phi^-_{B_2C_1}}\ket{\phi^+_{B_1C_2}}=\frac{1}{4}.
\end{eqnarray}
As done above, taking $x=1$ in the above expression, we get that using the fact  $B_0=C_0$ that
\begin{eqnarray}
    B_0= -(\I\otimes \sigma_z) B_0(\I\otimes \sigma_z).
\end{eqnarray}
Thus, $(\I\otimes \sigma_z)\ket{\xi}$ is an eigenvector of $B_0$ with eigenvalue $-1$ where $\sigma_z=\proj{0}+\proj{1}$. 

Similarly, considering all the other probabilities $p(a,b,c)$ for $a=3$, we have that
\begin{eqnarray}
    p(3,0,3)+p(3,3,0)+p(3,1,2)+p(3,2,1)=\frac{1}{16}\sum_{x=0}^3i^x\bra{\psi^-_{B_2C_1}}\bra{\phi^+_{B_1C_2}}B_0^x\otimes C_0^{x}\ket{\psi^-_{B_2C_1}}\ket{\phi^+_{B_1C_2}}=\frac{1}{4}.
\end{eqnarray}
As done above, taking $x=1$ in the above expression, we get that using the fact  $B_0=C_0$ that
\begin{eqnarray}
    B_0= -i(\I\otimes \sigma_z\sigma_x) B_0^{\dagger}(\I\otimes \sigma_z\sigma_x).
\end{eqnarray}
Thus, $(\I\otimes \sigma_z\sigma_x)\ket{\xi}$ is an eigenvector of $B_0$ with eigenvalue $-i$. 

Recalling that the sum of eigen-projectors of $B_0$ must sum up to identity, we get that
\begin{eqnarray}
    \proj{\xi}+(\I\otimes \sigma_x)\proj{\xi}(\I\otimes \sigma_x)+(\I\otimes \sigma_z)\proj{\xi}(\I\otimes \sigma_z)+(\I\otimes \sigma_z\sigma_x)\proj{\xi}(\I\otimes \sigma_z\sigma_x)=\I.
\end{eqnarray}
As $\ket{\xi}$ acts on two-qubit space, as done above in \eqref{P} it can written using Schmidt decomposition as $V\otimes\I\ket{\xi}=P\otimes\I\ket{\phi^+}$ and thus the above formula can be written as
\begin{eqnarray}
    P\otimes\I\left[\proj{\phi^+}+(\I\otimes \sigma_x)\proj{\phi^+}(\I\otimes \sigma_x)+(\I\otimes \sigma_z)\proj{\phi^+}(\I\otimes \sigma_z)+(\I\otimes \sigma_z\sigma_x)\proj{\phi^+}(\I\otimes \sigma_z\sigma_x)\right]P\otimes\I=\I.\nonumber\\
\end{eqnarray}

Notice that the sum inside the bracket is $\I$ and thus we get that $P^2=\I$ and thus Bob's and Charlie's measurement is the Bell-basis. This completes the proof.
\end{proof}

\section{Proofs for the $n-$ring scenario}

%Recalling the labelling scheme: $0 \rightarrow \alpha$, $1 \rightarrow \beta$, $2 \rightarrow -\alpha$ and $3 \rightarrow -\beta$ and then we construct the set $\mathcal{P}_{W_{n}} = \{\{a\}_n| \prod_{i=1}^{n}a_i = \alpha^{n-2j}\beta^{2j}\}$, with $j = 0,1,2,\dots$ $n/2$ or $(n-1)/2$ for even or odd $n$ respectively. Now the witness $W_n$ is defined as the sum of the probabilities of the elements of $\mathcal{P}_{W_{n}}$, that is, 
%\begin{eqnarray}\label{witness-na}
 %   W_n = \sum_{\mathcal{P}_{W_n}} p(\{a\}).
%\end{eqnarray}

\subsection{SOHS and quantum bounds of $n-$party swap steering witness}

The steering functional for $n-$networks is given by $W_n=\sum_{\{a\}_n}c_{\{a\}_n} p(\{a\}_n)$ where
\begin{equation}\label{witness-na1}
   c_{\{a\}_n}=\begin{cases}
       1\ \mathrm{if}\ a_1-a_2-\sum_{i=3}^n(-1)^{\sum_{k=2}^{i-1}a_k}a_i\ \mathrm{mod}\ 4 = 0\\
       0\ \mathrm{otherwise}.
   \end{cases}
\end{equation}
Let us find its SOHS bound.

\begin{fakt}
 Consider the $n-$party swap-steering functional $W_n$. The maximum value $\beta_{SOHS}$ that can be achieved using correlations that admit a SOHS model of $W_n$ is $\frac{1}{2}$.
\end{fakt}
\begin{proof}
    From the condition \eqref{witness-na1}, we have that the terms in the witness $W_n$ must satisfy $a_1 - a_2 - \sum_{i=3}^n (-1)^{\sum_{k=2}^{i-1}a_k}a_i\ \mathrm{mod}\ 4=0$. Observe that for every $\{a\}_{n-1/1}$, there exists a unique solution $a_1$, which implies that all combinations of $n-1$ strings $\{a\}_{n-1/1}$ appear exactly once in the summation. Thus, for the topologically robust version of the SOHS model \eqref{SOHS-nc}, the witness $W_n$ can be simplified as
    \begin{eqnarray}
         W_n \leq \max{\left[ \Tr{[N^{a_1}\rho_{\lambda_1}\otimes\rho_{\lambda_n}]} \right]} \sum_{\{a\}_{n-1/1}} \sum_{\lambda_1,\lambda_n} p(\lambda_1,\lambda_n) p(\{a\}_{n-1/1}|\lambda_n,\lambda_1) 
    \end{eqnarray}
    Using the definition of probabilities, $\sum_{\{a\}_{n-1/1}} p(\{a\}_{n-1/1}|\lambda_{1},\lambda_n) = 1$, for any $\lambda_1,\lambda_n$, we finally the obtain the bound as,
    \begin{eqnarray}
        W_n \leq \max{\left[ \Tr{[N^{a_1}\rho_{\lambda_1}\otimes\rho_{\lambda_n}]} \right]} = \frac{1}{2}.
    \end{eqnarray}
    This concludes the proof and also bounds the value of witness for SOHS models in n-ring networks.
\end{proof}

\begin{fakt}
    The maximal quantum value of the witness $W_n$ in the $n-$ring network with one trusted party is $1$.
\end{fakt}
\begin{proof}
For this proof we will use the following notation $\ket{\phi_+} \rightarrow \ket{0}$, $\ket{\psi_+} \rightarrow \ket{1}$, $\ket{\phi_-} \rightarrow \ket{2}$ and $\ket{\psi_-} \rightarrow \ket{3}$. The $n$ party state is then written as
\begin{eqnarray}\label{state-n}
    \ket{\psi_n} = \prod_{i}^{n}\ket{0}_{A_i^2,A_{i+1}^1} 
\end{eqnarray}
One can easily check the following properties of bell states given below.
\begin{eqnarray}\label{bell_swap}
    \ket{0}_{1,2}\ket{0}_{3,4} = \frac{\left(\ket{0}\ket{0} + \ket{1}\ket{1} + \ket{2}\ket{2} + \ket{3}\ket{3}\right)_{1,4,2,3}}{2} \nonumber \\ 
    \ket{0}_{1,2}\ket{1}_{3,4} = \frac{\left(\ket{1}\ket{0} + \ket{0}\ket{1} + \ket{2}\ket{3} + \ket{3}\ket{2}\right)_{1,4,2,3}}{2} \nonumber \\ 
    \ket{0}_{1,2}\ket{2}_{3,4} = \frac{\left(\ket{2}\ket{0} + \ket{0}\ket{2} - \ket{1}\ket{3} - \ket{3}\ket{1}\right)_{1,4,2,3}}{2} \nonumber \\ 
    \ket{0}_{1,2}\ket{3}_{3,4} = \frac{\left(\ket{3}\ket{0} - \ket{0}\ket{3} + \ket{1}\ket{2} - \ket{2}\ket{1}\right)_{1,4,2,3}}{2}
\end{eqnarray}
These equations are the key component of the proof. Notice how the state of the second subsystem determines the form of the permuted state in the same manner as the reverse coarse-graining scheme used to construct the witness. We will now prove that $p(\{a\}_n) = c_{\{a\}_n}/4^{n-1}$ for the $n$ partite state \eqref{state-n}. Using \eqref{bell_swap} we have,
\begin{eqnarray}
    \ket{\psi_n} = \left(\prod_{i}^{n-2}\ket{0}_{A_i^2,A_{i+1}^1}\right) \ket{0}_{A_{n-1}^2,A_{n}^1}\ket{0}_{A_n^2,A_{1}^1} = \left(\prod_{i}^{n-2} \ket{0}_{A_i^2,A_{i+1}^1} \right) \frac{1}{2} \left( \ket{0}\ket{0} + \ket{1}\ket{1} + \ket{2}\ket{2} + \ket{3}\ket{3} \right)_{A_{n-1}^2,A_{1}^1,A_n^1,A_{n}^2}. \quad
\end{eqnarray}
Suppose $n = 2$. The product on the left reduces to identity and the sum of states on the right produce the desired correlation $p(\{a\}_2) = c_{\{a\}_2}/4$. Repeating this step, we have
\begin{eqnarray}
    \ket{\psi_n} = \left(\prod_{i}^{n-3} \ket{0}_{A_i^2,A_{i+1}^1} \right) \frac{1}{4} \left( \ket{0}\ket{0}\ket{0} + \ket{0}\ket{1}\ket{1} + \ket{0}\ket{2}\ket{2} + \ket{0}\ket{3}\ket{3} + \ket{1}\ket{1}\ket{0} + \ket{1}\ket{0}\ket{1} + \ket{1}\ket{2}\ket{3} + \ket{1}\ket{3}\ket{2}+ \right. \nonumber \\ \left \ket{2}\ket{2}\ket{0} + \ket{2}\ket{0}\ket{2} - \ket{2}\ket{1}\ket{3} - \ket{2}\ket{3}\ket{1} + \ket{3}\ket{3}\ket{0} - \ket{3}\ket{0}\ket{3} + \ket{3}\ket{1}\ket{2} - \ket{3}\ket{2}\ket{1} \right)_{A_{n-2}^2,A_{1}^1,A_{n-1}^1,A_{n-1}^2,A_n^1,A_{n}^2}. \qquad
\end{eqnarray}
Again, suppose that $n = 3$. For the case of the triangle network, the product of the states on the left reduces to identity and the sum of states on the right reproduces the desired behavior $p(\{a\}_3) = c_{\{a\}_3}/16$. Further applications of the basis swap on the distributed state will respect the reverse coarse-graining scheme, which by definition, will produce terms of the n-ring witness. Each term of the n-ring witness will appear in the sum with equal probability $1/4^{n-1}$ (easily checked for the 2 and 3-party witness). This concludes the proof. 
\end{proof}

\subsection{Necessary conditions for swap steering}

\begin{fakt}
      If any one of the states generated by the sources or the measurement of any untrusted node is separable, then the maximum value that can be obtained of the swap-steering functional $W_n$ \eqref{witness-na1} is $\frac{1}{2}$.
\end{fakt}
\begin{proof}
We will prove the above statement by showing that one cannot obtain $n-$ring swap-steering correlations if any of the states or measurements are separable. For this purpose, we assume that any one of the sources $S_i$ distributes a separable quantum state $\rho_{A_{i}^2}\otimes\rho_{A_{i+1}^1}$. Consider now that all the parties except the trusted one obtain the outcomes $a_2,\ldots,a_n$. Consequently, the post-measured state at trusted node $A_1$ upto normalisation is given by
\begin{eqnarray}
    \Tr_{A_2\dots A_n}{\left[M^{a_2}_{A_2}\dots M^{a_n}_{A_n}\left(\rho_{A_1^2,A_2^1}\dots\rho_{A_i^2}\otimes\rho_{A_{i+1}^1}\dots\rho_{A_n^2,A_1^1}\right)\right]} = \Tr_{A_2\dots A_i}{\left[M^{a_2}_{A_2}\dots M^{a_i}_{A_i}\left(\rho_{A_1^2,A_2^1}\dots\rho_{A_{i-1}^2,A_i^1}\rho_{A_i^2}\right)\right]} \otimes \nonumber \\ \Tr_{A_{i+1}\dots A_n}{\left[M^{a_{i+1}}_{A_{i+1}}\dots M^{a_n}_{A_n}\left(\rho_{A_{i+1}^1}\rho_{A_{i+1}^2,A_{i+2}^1}\dots\rho_{A_{n}^2,A_1^1} \right)\right]} = \rho_{A_1^1}^{(a_{i+1},\dots ,a_n)}\otimes\rho_{A_1^2}^{(a_2,\dots,a_i)} \qquad\qquad\qquad\qquad
\end{eqnarray}
Similarly, let us consider one of the nodes $A_i$ to perform measurements such that for any $a_i$, $M_{A_i}^{a_i}\equiv M^{a_i}_{A_i^1} \otimes M^{a_i}_{A_i^2}$. State at trusted node $a_1$ post measurement at all other nodes is given by the unormalised state 
\begin{eqnarray}
    \Tr_{A_2\dots A_n}{\left[M^{a_2}_{A_2^1,A_2^2}\dots M^{a_i}_{A_i^1} \otimes M^{a_i}_{A_i^2} \dots M^{a_n}_{A_n^1,A_n^2}\left(\rho_{A_1^2,A_2^1}\dots\rho_{A_n^2,A_1^1}\right)\right]} = \Tr_{A_2\dots A_i^1}{\left[M^{a_2}_{A_2^1,A_2^2}\dots M^{a_i}_{A_i^1}\left(\rho_{A_1^2,A_2^1}\dots\rho_{A_{i-1}^2,A_i^1}\right)\right]} \otimes \nonumber \\ \Tr_{A_{i}^2\dots A_n}{\left[M^{a_{i}}_{A_{i}^2}\dots M^{a_n}_{A_n^1,A_n^1}\left(\rho_{A_i^2,A_{i+1}^1}\dots\rho_{A_{n}^2,A_1^1} \right)\right]} = \rho_{A_1^1}^{(a_{i}\dots a_n)}\otimes\rho_{A_1^2}^{(a_2 \dots a_i)} \qquad\qquad\qquad\qquad\qquad\quad
\end{eqnarray}
In both cases, $p(a_1|\{\rho\},a_2,a_3,\ldots,a_n) = \Tr{\left[N^{a_1}_{A_1^1,A_1^2} \rho_{A_1^1}\otimes\rho_{A_1^2}\right]} \leq 1/2$ for the Bell state projectors $N^{a_1}$ and set of quantum states distributed by the sources $\{\rho\}$. Using Baye's rule we have, $p(\{a\}_n|\{\rho\}) = p(a_1|\{\rho\},a_2,a_3,\ldots,a_n)p(a_2,a_3,\ldots,a_n|\{\rho\})$. Finally, we evaluate the witness $W_n$ as 
\begin{eqnarray}
    W_n \leq \max{\left[p(a_1|\{\rho\},a_2,a_3,\ldots,a_n)\right]} \sum_{a_2,a_3,\ldots,a_n} p(a_2,a_3,\ldots,a_n|\{\rho\}) = \frac{1}{2} \quad
\end{eqnarray}
Thus, if any state or measurement is separable then the value of the witness is bounded by $1/2$. By definition, these are not swap-steering correlations. This concludes the proof. 
\end{proof}

\subsection{Every bipartite entangled state is swap-steerable}

Consider an entangled state $\tilde{\rho}\in\mathbb{C}^{d_{A_1^1}}\otimes\mathbb{C}^{d_{A_1^2}}$ such that the corresponding entanglement witness is given by $W_{\tilde{\rho}}$ which can be expressed using the tomographically complete set of operators stated as $W_{\tilde{\rho}}=\sum_{s,t}\beta_{\tilde{\rho},s,t}\tau_s\otimes\omega_t$. Utilising this entanglement witness, the $n-$ring swap-steering inequalities is given by
\begin{eqnarray}\label{Wituniv}
    \mathcal{S}_{\tilde{\rho}}=-\sum_{s,t}\beta_{\tilde{\rho},s,t}p(0,0,\ldots,0|st)
\end{eqnarray}
where $p(0,0,\ldots,0|st)$ is the probability of obtaining outcome $0$ by all the parties given the input $s,t$ of the trusted party $A_1$.
\setcounter{fakt}{3}
\begin{fakt}
    Consider the $n-$ring swap-steering scenario  and the functional $\mathcal{S}_{\tilde{\rho}}$ \eqref{Wituniv}. The maximal value attainable of $\mathcal{S}_{\tilde{\rho}}$ using an SOHS model is $\beta_{SOHS}=0$. 
\end{fakt}
\begin{proof}
    Let us recall that for correlations admitting a SOHS model, we have that
    \begin{equation}
        p(\{a\}|st)= \sum_{\{\lambda\}}\prod_{j=1}^{n} p(\lambda_{j})\Tr{[N^{a_1}_{A_1}\rho_{\lambda_1}\otimes\rho_{\lambda_n}]}\prod_{i=2}^{n} p(a_{i}|\{\lambda_{i}\})
    \end{equation}
for all $\{a\}$. Consequently, we have from \eqref{Wituniv} that
\begin{eqnarray}\label{A21}
     \mathcal{S}_{\tilde{\rho}}&=&\sum_{\{\lambda\}}\prod_{j=1}^{n} p(\lambda_{j})\Gamma(\rho_{\lambda_1}\otimes\rho_{\lambda_n})\prod_{i=2}^{n} p(0,\ldots,0|\{\lambda_{i}\})\ \ 
\end{eqnarray}
where 
\begin{eqnarray}\label{gamma11}
    \Gamma(\rho_{\lambda_1}\otimes\rho_{\lambda_n})=-\sum_{s,t}\beta_{\tilde{\rho},s,t}p(0|st,\rho_{\lambda_1}\otimes\rho_{\lambda_n}).
\end{eqnarray}
Expanding the above formula \eqref{gamma11}, by recalling Alice's measurements $\mathcal{A}_{st}$, we obtain
\begin{eqnarray}
    \Gamma(\rho_{\lambda_1}\otimes\rho_{\lambda_n})&=& -\sum_{s,t}\beta_{\tilde{\rho},s,t}\Tr(\tau_s\otimes\omega_t\rho_{\lambda_1}\otimes\rho_{\lambda_n}) \nonumber\\&=&-\Tr(W_{\tilde{\rho}}\rho_{\lambda_1}\otimes\rho_{\lambda_n})\leq 0.
\end{eqnarray}
Thus, we have from \eqref{A21} that for correlations admitting a SOHS model
\begin{eqnarray}
     \mathcal{S}_{\tilde{\rho}}\leq0.
\end{eqnarray}
This concludes the proof.
\end{proof}

%%%%%%%%%%%%%%%%%%%%%
%\bibliography{ref.bib}

\end{document}